\newif\ifarxiv
\arxivtrue

\newif\ifdraftversion
\draftversionfalse

\ifarxiv
\documentclass[11pt]{article}
\usepackage{geometry}
\else
\documentclass[preprint,superscriptaddress,nofootinbib]{revtex4-1}
\fi

\ifarxiv
\usepackage{authblk}
\usepackage{setspace}

\fi

\usepackage{graphicx, amsfonts, amsmath, amsthm, enumerate,mathrsfs,xcolor,verbatim}
\ifarxiv
\IfFileExists{MinionPro.sty}
{\usepackage[lf]{MinionPro}

\usepackage{amsmath,amsthm,amsbsy}}
{\usepackage{times}
\usepackage{savesym} 
\usepackage{amsmath,amsthm,amsbsy}
\savesymbol{iint}
\savesymbol{openbox}
\usepackage{txfonts}

\DeclareSymbolFont{symbols}{OMS}{cmsy}{m}{n}
\restoresymbol{TXF}{iint}
\restoresymbol{TXF}{openbox}}
\fi

\ifdraftversion
\usepackage{lineno}
\let\oldalign\align
\let\oldendalign\endalign
\renewenvironment{align}
  {\linenomathNonumbers\oldalign}
  {\oldendalign\endlinenomath}
\fi

\usepackage{makeidx}
\makeindex

\ifarxiv
\definecolor{MyGreen}{rgb}{0,.6,.2}
\definecolor{MyDarkBlue}{rgb}{.1,.1,.75}
\usepackage[colorlinks=true, citecolor=MyGreen, linkcolor=MyDarkBlue]{hyperref}
\else
\usepackage{hyperref}
\fi

\newcounter{Figure}
\theoremstyle{plain}
\newtheorem{theorem}{Theorem}
\newtheorem{lemma}[theorem]{Lemma}
\newtheorem{proposition}[theorem]{Proposition}
\newtheorem{corollary}[theorem]{Corollary}

\theoremstyle{definition}
\newtheorem{definition}[theorem]{Definition}
\newtheorem{assumption}[theorem]{Assumption}

\numberwithin{theorem}{section}
\numberwithin{equation}{section}

\usepackage{slashed}
\newcommand{\Reals}{\mathbb R}
\newcommand{\Ints}{\mathbb Z}
\newcommand{\Nats}{\mathbb N}
\DeclareMathOperator{\Div}{\mathrm{div}}
\DeclareMathOperator{\im}{\mathrm{im}}
\DeclareMathOperator{\ck}{{\mathbb L}}

\newcommand{\be}{\begin{equation}}
\newcommand{\ee}{\end{equation}}
\newcommand{\bea}{\begin{eqnarray}}
\newcommand{\eea}{\end{eqnarray}}
\newcommand{\beas}{\begin{eqnarray*}}
\newcommand{\eeas}{\end{eqnarray*}}

\newcommand{\tr}{\text{tr} \,}

\newcounter{mnotecount}[section]
\let\oldmarginpar\marginpar
\renewcommand\marginpar[1]{\-\oldmarginpar[\raggedleft\footnotesize #1]%
{\raggedright\footnotesize #1}}

\begin{document}
\ifdraftversion
\linenumbers
\fi

\title{Asymptotically Euclidean Solutions of the Constraint Equations with Prescribed Asymptotics}

\ifarxiv
\author[1]{Lydia Bieri\thanks{lbieri@umich.edu}}
\affil[1]{Department of Mathematics, University of Michigan}

\author[2,3]{David Garfinkle\thanks{garfinkl@oakland.edu}}
\affil[2]{Department of Physics, Oakland University}
\affil[3]{Michigan Center for Theoretical Physics, Randall Laboratory of Physics, University of Michigan}

\author[4]{James Isenberg\thanks{isenberg@uoregon.edu}}
\affil[4]{Department of Mathematics, University of Oregon}

\author{David Maxwell\thanks{damaxwell@alaska.edu}}
\affil[5]{Department of Mathematical Sciences, University of Alaska Fairbanks}

\author[1]{James Wheeler\thanks{jcwheel@umich.edu}}

\else
\author{Lydia Bieri}
\email{lbieri@umich.edu}
\affiliation{Department of Mathematics, University of Michigan, Ann Arbor, MI 48109-1120, USA}

\author{David Garfinkle}
\email{garfinkl@oakland.edu}
\affiliation{Department of Physics, Oakland University,
Rochester, MI 48309, USA \\ 
and Michigan Center for Theoretical Physics, Randall Laboratory of Physics, University of Michigan, Ann Arbor, MI 48109-1120, USA}

\author{James Isenberg}
\email{isenberg@uoregon.edu}
\affiliation{Department of Mathematics, University of Oregon, Eugene OR 97403}

\author{David Maxwell}
\affiliation{Department of Mathematical Sciences, University of Alaska Fairbanks, Fairbanks, AK, 99775, USA}
\email{damaxwell@alaska.edu} 

\author{James Wheeler}
\email{jcwheel@umich.edu}
\affiliation{Department of Mathematics, University of Michigan, Ann Arbor, MI 48109-1120, USA}
\fi

\date{\today}

\ifarxiv
\maketitle
\fi

\begin{abstract} 
We demonstrate that in constructing asymptotically flat vacuum initial data sets in General Relativity via the conformal method, certain asymptotic structures may be prescribed a priori through the specified seed data, including the Arnowitt–Deser–Misner (ADM) momentum components, the leading- and next-to-leading-order decay rates, and the anisotropy in the metric's mass term, yielding a recipe to construct initial data sets with desired asymptotics. We numerically construct a simple explicit example of an initial data set, with stronger asymptotics than have been obtained in previous work, such that the evolution of this initial data set does not exhibit the conjectured antipodal symmetry between future and past null infinity.
\end{abstract} 

\ifarxiv
\newpage
\else
\maketitle
\fi

\tableofcontents

\section{Introduction}
\label{intro}

General relativity is our foremost theory of gravity on classical scales, and its predictive power largely stems from the well-posedness of the initial value problem formulation: Given a global ``instant in time" characterized by a complete Riemannian 3-manifold $(\Sigma, \tilde g)$ equipped with a symmetric $(0,2)$-tensor $\widetilde K_{ab}$ subject to the Einstein constraint equations \eqref{eqn:hamiltonian-gen}-\eqref{eqn:momentum-gen}, one seeks a globally hyperbolic Lorentzian 4-manifold $(\widehat M, \widehat g)$ satisfying Einstein's gravitational field equation into which $\Sigma$ embeds as a Cauchy surface with the induced metric $\tilde g$ and the second fundamental form $\widetilde K_{ab}$. 
Since Choquet-Bruhat's seminal result, proven over half a century ago \cite{choquet1952,choquet1969global}, establishing the existence of such a (maximal) spacetime $(\widehat M, \widehat g)$ under general hypotheses, the mathematical relativity community has been broadly interested in characterizing the geometric structure of $(\widehat M, \widehat g)$ satisfying various conditions on the initial data set $(\Sigma, \tilde g, \widetilde K)$, with particular emphasis on the {\it asymptotically flat} case modeling an isolated system. 
Features of particular interest include the existence of past and future null infinity $\mathscr I^\pm$ with particular prescribed curvature behaviors there, and much work has been done on establishing the stability of these behaviors under asymptotically flat perturbations of initial data sets extracted from fundamental explicit solutions such as the Minkowski, the Schwarzschild, and the Kerr spacetimes. 
Such works invoke different notions of asymptotic flatness to different effects, and there is some interest in the question of how stringent one's notion of asymptotic flatness on $(\Sigma,\tilde g, \widetilde K)$ must be to ensure that desirable physical features are realized within $(\widehat M, \widehat g)$.

The first work which established the broad existence of solutions to the initial value problem for Einstein's equation admitting a complete notion of infinity, outside of the few known explicit spacetimes, was Christodoulou and Klainerman's proof of the nonlinear stability of Minkowski space \cite{sta}. 
They proved the stability of key geometric features of the maximal vacuum Cauchy development of sufficiently ``small" initial data sets that are asymptotically flat in the sense that there exists a coordinate system $(x_1,x_2,x_3)$ on the complement of a compact set in $\Sigma$ in which\footnote{For $f \in C^m(\mathbb R^n)$, we define $f = o_m(r^\delta)$ provided that $D^\alpha f = o(r^{\delta-|\alpha|})$ as $r \to \infty$ for any multi-index $\alpha$ of order $|\alpha| \leq m$.}
\begin{align}
\tilde g_{ij} & = \left(1+\frac{2M}{r}\right)\delta_{ij} + o_4(r^{-3/2}), \label{eqn:typeckg} \\
\widetilde K_{ij} & = o_3(r^{-5/2}). \label{eqn:typeckk}
\end{align}
We define these falloff conditions as type (CK). This conclusion was later generalized by Bieri \cite{lydia1,lydia2} to initial data sets with the following relaxed falloff conditions (type (B)):
\begin{align}
\tilde g_{ij} & = \delta_{ij} + o_3(r^{-1/2}), \label{eqn:typebg} \\
\widetilde K_{ij} & = o_2(r^{-3/2}). \label{eqn:typebk}
\end{align}
In each case, one has different degrees of control on the induced falloff rates of various curvature components as one approaches $\mathscr I^\pm$. As also investigated by Bieri \cite{lydia4,lydia14,lydia5}, different control still is afforded by the following intermediate falloff conditions (type (A)):
\begin{align}
\tilde g_{ij} & = \delta_{ij} + h_{ij} + o_3(r^{-3/2}), \label{eqn:typeag} \\
\widetilde K_{ij} & = o_2(r^{-5/2}), \label{eqn:typeak}
\end{align}
where each component of $h_{ij}$ is homogeneous of degree $-1$, representing an anisotropic mass. 

To better understand the asymptotic behavior of the geometry of spacetimes compatible with these three falloff conditions, and to assess the status of physically-motivated conjectures associated to that geometry, it is interesting and necessary to construct examples of constraint-satisfying initial data sets in each of the classes (CK), (B), and (A), requiring the ability to {\it prescribe} the data's falloff rate.
While the ultimate act of construction is a numerical task, it is vital to first develop a mathematical blueprint of the procedure underpinned by theorems guaranteeing its efficacy. 
A subset of the present authors \cite{bieri_brill_2025} recently provided such a blueprint in the rather restrictive setting of Brill wave initial data sets, requiring both azimuthal and time symmetry. These results, however, were insufficient to guarantee the construction of the more delicate type (A) data. That work numerically constructed an initial data set of type (B) whose evolution cannot admit the spacetime symmetry suggested by Strominger's antipodal conjecture in \cite{strominger}, which is a topic of interest in the literature \cite{symmprabhuetal,symmmagdyetal,lz1}. The symmetry in question involves a particular component of the electric part of the Weyl curvature ($\Psi_2$ in the Newman-Penrose formalism, or $\rho$ in Christodoulou and Klainerman's notation) evaluated asymptotically along $\mathscr I^+$ as compared to $\mathscr I^-$: If one evaluates this component's limit on each of $\mathscr I^\pm$ and proceeds along each to spatial infinity, the two resulting functions on the sphere should agree up to composition with the antipodal map on the sphere, $p \mapsto -p$. %The physical conjecture actually pertains to certain charges at null infinity which only exist if this curvature component decays sufficiently rapidly, however, which will not occur in the evolution of type (B) initial data sets which indeed decay strictly slower than $O(r^{-1})$, so the example found in \cite{bieri_brill_2025} was not fully satisfactory.

In this work, we significantly generalize the analytical results of \cite{bieri_brill_2025}, both to do away with any symmetry assumptions and to handle type (A) data, establishing a sequence of results guaranteeing that one can control the falloff rates and the Arnowitt–Deser–Misner (ADM) momenta of initial data sets generated via the widely utilized conformal method. 
We use the same numerical method as in \cite{bieri_brill_2025} to construct an initial data set, which our new analytical results guarantee is type (A), whose evolution cannot admit the symmetry suggested by the conjecture of \cite{strominger}, thereby restricting the class of spacetimes which admit this symmetry. 
Since this symmetry holds trivially in spacetimes which are evolved from type (CK) data, our results indicate that the domain on which the symmetry holds nontrivially must be somewhat small. 
For simplicity of presentation, we work with $\Sigma = \mathbb R^n$, but our analytical results generalize straightforwardly to all asymptotically Euclidean manifolds, even those with multiple ends. Of course, one takes $n = 3$ in standard general relativity (as above), but we have found it interesting and instructive to generalize the mathematical results presented here to arbitrary dimension $n \geq 3$.

We work with finite regularity measured in Sobolev scales, and in particular with weighted Sobolev spaces. This common setting for building asymptotically Euclidean initial data usually focuses on decay rates within the isomorphism range of the applicable elliptic operators (e.g. \cite{bruhat-isenberg}, \cite{maxwell_rough_2006}).  While this approach suffices for establishing existence and uniqueness, the slow decay rates lose asymptotic information. We work instead with faster decay rates that expose the asymptotic structure, and with particular care at the critical transition rates where the Fredholm index of the operators changes (see Proposition \ref{prop:flat-P-threshold} and Lemma \ref{lem:borderline}).  Although we work exclusively in the context of constant mean curvature (CMC) solutions of the constraints, we anticipate that the tools developed here could be applied to perturbative constructions of near-CMC initial data using implicit function theorem techniques as in \cite{bruhat-isenberg}.

We remark that Dain and Friedrich published a work in 2001 \cite{dain2001asymptotically} which has a fair bit of {\it conceptual} overlap with (and a remarkably similar title to) the present work. 
They also proved a number of interesting results allowing control of the asymptotic features of initial data sets constructed via the conformal method, including the prescription of a leading-order Euclidean transverse traceless tensor in $\widetilde K_{ij}$, which allows one to control the ADM momentum (cf.\@ our Theorem \ref{thm:summary_a}). 
However, there is little {\it technical} overlap of \cite{dain2001asymptotically} with the present work, either regarding precisely what is proved or regarding the analytical techniques employed. 
Specifically, Dain and Friedrich worked (appropriately to their objectives) in the restricted setting of initial data sets which admit a conformal compactification at spatial infinity, so they were interested in obtaining asymptotics which amount to analyticity at spatial infinity, requiring that $\tilde g_{ij}$ and $\widetilde K_{ij}$ admit a full power series expansion in integer powers of $1/r$.
The present work seeks to provide a simple procedure for constructing initial data sets in a much broader category, with an emphasis on controlling the structure of the leading-order terms (especially in $\tilde g_{ij}$) and controlling the precise (non-integral) decay rate of subleading terms, and making no demand for a full power series expansion. 
Broadly, we are interested in the construction of initial data sets conforming to various notions of asymptotic flatness for the ultimate purpose of probing the extent to which physically interesting features of spacetime depend on this notion, but the results in \cite{dain2001asymptotically} are not amenable to handling a variety of falloff behaviors.
In particular, the results of \cite{dain2001asymptotically} allow one to construct highly
specialized cases of type (CK) initial data sets, but they cannot be used for the construction of type (A) initial data sets which are not type (CK), which is among our chief concerns.

In a similar vein, in 1996 Beig and O'Murchadha used spherical harmonic expansions to investigate the asymptotic behavior of solutions to the CMC vacuum momentum constraint equation \eqref{eqn:momentum} \cite{beig_momentum_1996}. Working in the smooth setting in dimension $n = 3$, they fully characterized the asymptotic forms and decay rates of solutions to the momentum constraint with seed data having given arbitrarily fast, non-exceptional falloff compatible with type (CK) data. This characterization was made up to a collection of constant ``multipole moments", including the ADM momentum components, which can be prescribed via an appropriate choice of seed data (but cannot be anticipated for an arbitrary choice of seed data). We build upon Beig and O'Murchadha's work as follows: we extend some of their conclusions to lower regularity and arbitrary dimension; we establish that similar conclusions hold for solutions to the Hamiltonian constraint equation  \eqref{eqn:hamiltonian} (and hence one's entire initial data set); we show that the ADM momentum components can be computed directly from the seed data; and we investigate the asymptotics of solutions if the seed data has the exceptional ``threshold" falloff $O(r^{2-n})$, allowing the construction of type (A) data which is not (CK).

The remainder of this work is organized as follows: In Section \ref{sec:prelim}, we establish our notational
conventions and discuss the procedure for carrying out the conformal method.
To proceed with obtaining solutions of the constraint equations using the conformal
method, one must solve the ``conformal constraint equations"
 \eqref{eqn:lichnerowicz} and \eqref{eqn:conformalmomentum} for the 
positive scalar $\varphi$ and the vector field $W$. We note in that section that, so long as we
presume that the mean curvature $\tr_{\tilde g}(\widetilde K)$ for the initial data set is zero, these
conformal constraint equations semi-decouple. In Section \ref{sec:momentum}, we discuss the conformal
momentum constraint equation \eqref{eqn:conformalmomentum} by analyzing the mapping properties of the vector Laplacian operator, and we show how to choose the ``seed data" so that the extrinsic curvature of the solution has various falloff properties. 
We show there that the ADM momentum of the ultimate initial data set may be prescribed a priori through the choice of the seed data.
In Section \ref{sec:hamiltonian}, we discuss the Hamiltonian constraint equation \eqref{eqn:lichnerowicz} in a manner closely analogous to the discussion of the momentum constraint, showing that the seed data can be chosen so that the conformal factor $\varphi$ implements the desired falloff properties. 
In Section \ref{sec:constructions}, we summarize the implications of our results for the construction of full initial data sets satisfying the desired falloff rates-- those of type (CK), (A), and (B) data--, collated into Theorems \ref{thm:summary_b}, \ref{thm:summary_ck}, and \ref{thm:summary_a}. In Section \ref{sec:numerics}, we produce a numerical example of type (A) data whose evolution cannot admit the symmetry suggested by the antipodal matching conjecture in \cite{strominger}. We briefly summarize and conclude in Section \ref{sec:conclusion}.
Appendix \ref{app:vec-fredholm} presents a largely self-contained Fredholm theory of vector Laplacians, analogous to results for the scalar Laplacian-type operators in \cite{bartnik1}.

\section{Preliminaries and the Conformal Method}
\label{sec:prelim}

The problem of constructing initial data sets $(\Sigma,\tilde g, \widetilde K)$ is more analytically challenging than one might like primarily due to the constraint equations, which are geometrically imposed relations between the geometry of $(\widehat M, \widehat g)$ and $(\Sigma, \tilde g, \widetilde K)$ which read (in vacuum)
\begin{align}
R(\tilde g) - |\widetilde K|_{\tilde g}^2 + \text{tr}_{\tilde g}(\widetilde K)^2 & = 0,
\label{eqn:hamiltonian-gen}
\\ \text{div}_{\tilde g}(\widetilde K) - d(\text{tr}_{\tilde g}(\widetilde K)) & = 0.
\label{eqn:momentum-gen}
\end{align}
The first of these is the {\it Hamiltonian constraint}, and the second is the {\it momentum constraint}. There is a broad literature on the study of the constraint equations, spanning from methods of constructing solutions, either from scratch or from existing solutions, to properties that solutions enjoy, such as the positive mass theorem and the Riemannian Penrose inequality. We are interested in the former, and we seek to construct initial data sets corresponding to a {\it maximal} time slice in $\widehat M$, wherein $\tr_{\tilde g}(\widetilde K) = 0$. The vacuum constraint equations then take the following form:
\begin{align}
R(\tilde g) & = |\widetilde K|_{\tilde g}^2, \label{eqn:hamiltonian}
\\ \text{div}_{\tilde g}(\widetilde K) & = 0 \label{eqn:momentum}.
\end{align}
We pose these equations on $\mathbb R^n$, to be solved for $\tilde g$ and $\widetilde K$. These are underdetermined, as they impose only $n+2$ conditions (including $\tr_{\tilde g}(\widetilde K) = 0$) on the two symmetric (0,2)-tensors $\tilde g$ and $\widetilde K$. To determine a parameterized set of solutions, one must supplement these conditions, effectively specifying part of $\tilde g$ and $\widetilde K$ and leaving the rest of their content to be determined by the constraints. 

A historically powerful approach in the analysis of these equations, known as the {\it conformal method} \cite{bruhat-york,lichnerowicz1944}, is to specify the conformal class of $\tilde g$, amounting to stipulating an arbitrary ``unphysical" metric $g$, as well as a traceless and symmetric $(0,2)$-tensor $A_{ab}$  (the {\it conformal velocity}) and a scalar function $N = 1+\delta N : \mathbb R^n \to \mathbb R_+$ (the {\it lapse function}). Once this {\it seed data} comprised of  $g$, $A_{ab}$, and $\delta N$ has been specified on $\Reals^n$, one sets
\[
K_{ab} := \frac{1}{2N}\left( A_{ab} - (\mathbb L_g W)_{ab} \right)
\]
and $q_n := \frac{2n}{n-2}$ and proceeds to solve the PDE system
\begin{align}
-(q_n+2)\Delta_g \varphi + R(g) \varphi & = \varphi^{-q_n-1} |K|_g^2, \label{eqn:lichnerowicz}\\
\text{div}_g\left( \frac{1}{2N} (\mathbb L_g W)_{ab} \right) & = \text{div}_g\left( \frac{1}{2N}  A_{ab} \right) \label{eqn:conformalmomentum}
\end{align}
for the vector field $W^a$ and the positive scalar field $\varphi$. 
Here, $\mathbb L_g$ is the conformal Killing operator, which is defined to be the following symmetrized and traceless covariant derivative: 
\[
\mathbb (\mathbb L_g W)_{ab} := \nabla_a W_b + \nabla_b W_a - \frac{2}{n} g_{ab} \nabla_c W^c.
\]
Once $W^a$ and $\varphi$ have been obtained by solving equations \eqref{eqn:lichnerowicz} and \eqref{eqn:conformalmomentum}, one sets 
\begin{align}
\tilde g & := \varphi^{q_n-2} g, \notag
\\ \widetilde K_{ab} & := \varphi^{-2} K_{ab}. \notag
\end{align}
Using the conformal divergence identity $\text{div}_{\tilde g}(\varphi^{-2} K) = \varphi^{-q_n} \text{div}_{g}(K)$ for a traceless and symmetric $(0,2)$-tensor $K_{ab}$ as well as the well-known conformal transformation of scalar curvature, one readily confirms that this choice of the physical metric $\tilde g$ and the second fundamental form $\widetilde K_{ab}$ solves the constraint equations \eqref{eqn:hamiltonian}-\eqref{eqn:momentum}.

A key feature of the conformal method for maximal data is that the conformal constraint equations \eqref{eqn:lichnerowicz}-\eqref{eqn:conformalmomentum} are semi-decoupled in the sense that \eqref{eqn:conformalmomentum} can be solved for $W^a$ independently of $\varphi$, and then equation \eqref{eqn:lichnerowicz}– the Lichnerowicz equation– can be solved for $\varphi$. We note that in the Einstein vacuum case discussed here, these equations admit unique solutions for each choice of the seed data in appropriate function spaces (discussed below) \cite{maxwell_solutions_2005}.

We are interested in prescribing the falloff rates of $\tilde g$ and $\widetilde K$ by means of building them into the seed data $g$, $A_{ab}$, and $\delta N$. To leverage analytical control over the decay rates of these quantities, we operate in weighted Sobolev spaces, which we now define:

\begin{definition}\label{def:weighted-basic}
    Let $k\in \Ints_{\ge 0}$, $1<p<\infty$ and $\delta\in \Reals$.
    The \textbf{weighted Sobolev space} $W^{k,p}_\delta$ consists of the functions 
    $u:\Reals^n\to\Reals$ such that 
\[
    ||u||_{W^{k,p}_\delta}^p := \sum_{|\alpha|\le k}\int_{\Reals^n} \left<x\right>^{-n - p\delta} |\partial_{\alpha} u|^p\; dx < \infty,
\]
where $\left<x\right> = \sqrt{1+|x|^2}$.
The analogous norm for tensors on $\Reals^n$ is defined 
component-wise, and we use the same notation $W^{k,p}_\delta$ for spaces of tensor-valued functions,
leaving the tensor bundle implicit. We also introduce the following notation:
\[ W^{k,p}_{\delta^+} := \bigcap_{\epsilon > 0} W^{k,p}_{\delta+\epsilon}, \qquad \qquad W^{k,p}_{\delta^-} := \bigcup_{\epsilon > 0} W^{k,p}_{\delta-\epsilon} \]
\end{definition}

Intuitively (and literally, provided $k > n/p$), the role of the weight $\delta$ is to impose that these spaces are comprised of functions that decay more quickly than $|x|^\delta$ at infinity, being $o(|x|^\delta)$, with control of (distributional) derivatives up to order $k$. The ``plus" space $W^{k,p}_{\delta^+}$ further allows for terms decaying at exactly the rate $|x|^\delta$, as well as for log terms such as $\left<x\right>^\delta \log(\left<x\right>)$, and the ``minus" space $W^{k,p}_{\delta^-}$ ensures that one has some power law decay strictly faster than $|x|^\delta$. A weight $\delta$ is defined to be {\it exceptional} if it is an integer and either $\delta\le 2-n$ or $\delta\ge0$, and {\it non-exceptional} otherwise. Exceptional weights are precisely the possible growth rates of homogeneous harmonic functions on flat space-- these act as barriers at which the analytical properties of our elliptic operators acting on these spaces, such as their kernel and cokernel dimensions, may change. Extension of Definition \ref{def:weighted-basic} to negative values of $k \in \mathbb Z$ is necessary for complete coverage of the results obtained herein, but such technicalities are reserved for discussion in the appendix: this classical definition is sufficient to parse the main body of this work. We now define asymptotically Euclidean metrics:

\begin{definition} A metric $g_{ab}\in W^{k,p}_{\mathrm{loc}}$ 
    with $1<p<\infty$ and $k\in \Nats$ satisfying $k>n/p$ is \textbf{asymptotically Euclidean 
    of class} $W^{k,p}_{\tau}$ for some $\tau<0$ if
    \[
    g_{ab}- \delta_{ab} \in W^{k,p}_{\tau}.
    \] 
\end{definition}

We note that the condition $k>n/p$ ensures that $g_{ab}$ is H\"older continuous
and, since $\tau < 0$, that it converges uniformly to the Euclidean metric $\delta_{ab}$ as $x\to\infty$. In keeping with this definition, we work with metrics and lapses in the following category:

\begin{assumption}\label{assumption:main}
    The metric $g_{ab}$ and the lapse $N$ satisfy the following two conditions for some $\tau<0$,  $k\in\Nats$, and $1<p<\infty$ with $k > 1+n/p$:
    \begin{itemize}
        \item $g_{ab}$ is asymptotically Euclidean of class $W^{k,p}_{\tau}$.
        \item $N$ is a positive function with $\delta N = N-1 \in W^{k,p}_{\tau}$
    \end{itemize}
\end{assumption}

We remark that the core results of this work are true under the relaxed hypothesis $k > n/p$, provided that one additionally assumes that $g$ admits no nontrivial conformal Killing fields (vector fields in the kernel of $\ck$) of class $W^{k,p}_{\delta}$ for any $\delta < 0$, i.e.\@ which vanish at infinity\footnote{It is a reasonable conjecture that this additional assumption is implied by $k > n/p$, but a proof remains elusive.}. The following essential multiplication lemma is an immediate consequence of the more general Lemma \ref{lem:mult} discussed in Appendix \ref{app:vec-fredholm}, and we note that it suffices for all of our work here.
\begin{lemma}\label{lem:mult-basic}
	Suppose $1<p<\infty$, $k_1,k_2,j\in\Ints_{\geq 0}$ and $\delta_1, \delta_2 \in\Reals$.
Pointwise multiplication of functions in $C^\infty_c(\mathbb R^n)$  extends to a continuous bilinear map 
$W^{k_1,p}_{\delta_1}\times W^{k_2,p}_{\delta_2}
\rightarrow W^{j,p}_{\delta_1 + \delta_2}$ if
\begin{align}
    j & \le \min \{k_1,k_2\}, \notag \\
%\frac{2}{p} -\frac{k_1+k_2}{n} &< \frac{1}{p} - \frac{j}{n}. \\
    j & < k_1+k_2 - \frac{n}{p}. \notag
\end{align}
In particular, we may choose the optimal value $j = \min\{k_1,k_2\}$ so long as $\max\{k_1, k_2\} > n/p$.
\end{lemma}
\qed

\section{The Momentum Constraint}
\label{sec:momentum}

The asymptotics of a second fundamental form constructed using the conformal method follow from the mapping properties of the 
\textbf{lapse-weighted vector Laplacian} $P_{g,N} := \Div_g (\frac{1}{2N} \ck)$ appearing in equation \eqref{eqn:conformalmomentum} along with 
the asymptotic structure of its source term, the conformal velocity $A_{ab}$.
Indeed, if we set $Z_a := \nabla^b \left( \frac{1}{2N} A_{ab} \right)$, then equation (\ref{eqn:conformalmomentum}) becomes the inhomogenenous linear equation
\be \label{eqn:pmomentum}
(P_{g,N} W)_a = Z_a
\ee
which is the main focus of this section.

A routine computation using Lemma \ref{lem:mult-basic}, under Assumption \ref{assumption:main}, yields that for a smooth vector field $X^a$ one has
\be \label{eqn:pdiff}
(P_{g,N} X)_a = (\overline P X)_a + \sum_{|\alpha|\le 2} B_{ab}^\alpha \partial_\alpha X^b,
\ee
where $\overline P := P_{\overline g, 1} = \frac{1}{2} \overline \Div (\overline \ck)$ is the Euclidean vector Laplacian (we note that we have an extra factor of $1/2$ relative to some other works),
associated with the Euclidean metric $\overline g_{ab}=\delta_{ab}$, and where each $B_{ab}^{\alpha} \in W^{k-2+|\alpha|,p}_{\tau-2+|\alpha|}$ depends on $g$ and $N$. Since $\tau-2+|\alpha|<0$ for each of these coefficents, the mapping properties of $P_{g,N}$ can be deduced from those of $\overline P$. Appendix \ref{app:vec-fredholm} contains these details, which imply the following:

\begin{proposition} \label{prop:vec-lap-properties}
    Suppose that $g_{ab}$ and $N$ satisfy Assumption \ref{assumption:main} with parameters $k$, $p$
    and $\tau$.  Then for any $\delta\in\Reals$, $P_{g,N}: W^{k,p}_\delta \to W^{k-2,p}_{\delta-2}$
    is continuous. If $\delta$ is non-exceptional, then
    \begin{enumerate}[(a)]
    \item $P_{g,N}$ acting between these spaces is Fredholm, \label{part:fredholm}
    \item its Fredholm index equals that of the Euclidean vector Laplacian $\overline{P}$
    acting between the same spaces, \label{part:index}
    \item if $\delta<0$ then its kernel is trivial,\label{part:no-kernel}
    \item if $\delta>2-n$ then it is surjective,\label{part:no-cokernel}
    \item if $2-n<\delta<0$ then it is an isomorphism,\label{part:iso}
    \item for any $V_a\in W^{k-2,p}_{\delta-2}$,
    the equation 
    \[
        (P_{g,N} X)_a = V_a
    \]
    is solvable for $X^a \in W^{k,p}_\delta$ if and only if
    $\int_{\Reals^n} V^a k_a \, dV_g = 0$ for all vector fields $k^a$ in the kernel of
    $P_{g,N}$ acting on $W^{k,p}_{2-n-\delta}$.\label{part:P-solvable}
    \end{enumerate}
\end{proposition}
\begin{proof}
The continuity of $P_{g,N}$ follows from Lemma \ref{lem:mult-basic},
and parts (\ref{part:fredholm}), (\ref{part:index}), and (\ref{part:P-solvable})
follow from Proposition \ref{prop:fredholm}.  Since $k>n/p+1$, Theorem 6.4 of \cite{maxwell_solutions_2005} implies 
that $g_{ab}$ admits no nontrivial conformal Killing fields in $W^{k,p}_{\delta}$ for 
any $\delta<0$, which is equivalent to part (\ref{part:no-kernel}). Part (\ref{part:no-cokernel})
now follows from parts (\ref{part:no-kernel}) and (\ref{part:P-solvable}),
and part (\ref{part:iso}) is immediate from parts (\ref{part:no-kernel}) and (\ref{part:no-cokernel}).
\\
\end{proof}

The isomorphism range $2-n< \delta < 0$ appearing in Proposition \ref{prop:vec-lap-properties} is well studied in the literature and 
is a routine tool used to find solutions of the momentum constraint (e.g., \cite{bruhat-isenberg}, \cite{maxwell_rough_2006}).
Understanding the asymptotics of these solutions, however, requires a finer analysis of the next decay range $1-n< \delta < 2-n$.  
Part \eqref{part:P-solvable} relates solvability in this range to elements of the kernel of $P_{g,N}$ that are bounded at infinity, and we have the following characterization
of this kernel:

\begin{proposition}\label{prop:kernel}
    Suppose that $g_{ab}$ and $N$ satisfy Assumption \ref{assumption:main} with parameters $k$, $p$
    and $\tau$.  If $0<\delta<1$, then the kernel of $P_{g,N}$ acting on $W^{k,p}_\delta$ 
    has dimension $n$.
    Moreover, there is a basis for this kernel consisting of vector fields 
    $k_{(i)}^a$, $i=1,\ldots,n$ satisfying\footnote{We use parenthetical indices to denote simple labels, 
    distinguished from the usual tensorial indices.}
    \[
    k_{(i)}^a - e_{(i)}^a \in W^{k,p}_\tau \cup W^{k,p}_{(2-n)^+}
    \]
    where $e_{(i)}$ denotes the $i^{\text{th}}$ standard coordinate basis vector on $\mathbb R^n$.
\end{proposition}
\begin{proof}
    Proposition \ref{prop:vec-lap-properties} implies that the maps $P_{g,N}$ and $\overline P$ acting on $W^{k,p}_\delta$ 
    are both Fredholm with the same index, and that they are both surjective (since $\delta>2-n$). The dimensions of their kernels therefore agree.
    Since Proposition \ref{prop:flatcase} indicates that the kernel of $\overline P$
    consists only of polynomials, which must be constant since $0<\delta<1$, 
    it follows that $\{e_{(i)}^a\}_{i=1}^n$ comprises a basis of this kernel, and
    it has dimension $n$.

    Equation (\ref{eqn:pdiff}) and Lemma \ref{lem:mult-basic} yield
    \[
    (P_{g,N} e_{(i)})_a = ((P_{g,N}-\overline{P}) e_{(i)})_a \in W^{k-2,p}_{\tau-2}.
    \]
    We pick $\eta$ satisfying $\eta\ge \tau$ and $2-n<\eta<0$.
    Since $W^{k-2,p}_{\tau-2} \subset W^{k-2,p}_{\eta-2}$,
    and since $P_{g,N}$ acting on $W^{k,p}_{\eta}$ is an isomorphism, 
    we can find $U_{(i)}^a\in W^{k,p}_{\eta}$ satisfying 
    \[
    (P_{g,N} U_{(i)})_a = - (P_{g,N} e_{(i)})_a.
    \]
    Uniqueness of the solution implies that the choice of $\eta$ 
    is inessential and $U_{(i)}^a\in W^{k,p}_{\tau} \cup W^{k,p}_{(2-n)^+}$.
    
    By construction $k_{(i)}^a := e_{(i)}^a + U_{(i)}^a \in \ker P_{g,N}$. Moreover,
    since the perturbations $U_{(i)}^a$ vanish at infinity,
    it follows that the $n$ vector fields $k_{(i)}^a$ are linearly independent, forming a basis for $\ker P_{g,N}$.  
    \\
\end{proof}

Now suppose $Z^a \in W^{k-2,p}_{\delta-2}$ with fast decay
$1-n<\delta<2-n$.  Although Proposition \ref{prop:vec-lap-properties}
ensures equation \eqref{eqn:pmomentum} can be solved 
for $W^a\in W^{k,p}_{(2-n)^+}$, the solution 
does not have faster decay in general.  Instead, 
Proposition \ref{prop:vec-lap-properties}(\ref{part:P-solvable}) indicates that 
$W^a\in W^{k,p}_\delta$ only if each of the integrals
\be \label{eqn:momentumints}
\mathcal O_{i}(Z_a) := \int_{\mathbb R^n} Z_a k_{(i)}^a dV_g
\ee
vanishes.  In fact, these obstructions turn out to be closely related to the ADM momentum, which we now examine.

Consider a $C^1$ symmetric, traceless tensor $S_{ab}$ with 
$O(r^{1-n})$ growth. Assuming units with $8\pi G=1$ 
the ADM momentum of $S_{ab}$ has components
\be \label{eqn:admmomentum} 
\mathcal P_i(S_{ab}) := \lim_{R\to\infty} \int_{\partial B_R} S_{ab} n^a e_{(i)}^b\; d\overline{A},
\ee
with $n^a := x^a/r$ the Euclidean radial unit vector, so long as the limit exists. We can express
this quantity equally well in terms of the metric $g_{ab}$, its area element $dA$, its normal
vector $\nu^a$ to the spheres, and the kernel basis $k^a_{(i)}$.  
Because $n^a-\nu^a$, $dA-d\overline{A}$ and $k^a_{(i)}-e^a_{(i)}$ each decays at the rate $O(r^\tau)$, we find
\be \label{eqn:admmomentum-2} 
\mathcal P_i(S_{ab}) := \lim_{R\to\infty} \int_{\partial B_R} S_{ab} \nu^a k_{(i)}^b\; dA.
\ee
The divergence theorem then yields
\begin{align*}
\int_{\partial B_R} S_{ab} \nu^a k_{(i)}^b\; dA & =
\int_{B_R} \nabla^a (S_{ab} k_{(i)}^b)\; dA \nonumber \\
& = 
\int_{B_R}  S_{ab} (\nabla^a k_{(i)}^b)\; dA
+ 
\int_{B_R} (\nabla^a S_{ab}) k_{(i)}^b\; dA.
\end{align*}
If $(\nabla^a S_{ab}) k_{(i)}^b$ and $S_{ab} (\nabla^a k_{(i)}^b)$ are
$L^1$ functions, we can take a limit in $R$, and we define
\begin{align*}
\mathcal G_i(S_{ab}) & := \int_{\Reals^n} S_{ab} \nabla^a k_{(i)}^b dV_g,\\
\mathcal R_i(S_{ab}) & := \int_{\Reals^n} (\nabla^a S_{ab}) k_{(i)}^b dV_g.
\end{align*}
Note that the quantity $\mathcal R_i$ defined here for symmetric, trace-free $(0,2)$-tensors
is related to the obstruction values $\mathcal O_i$ for covectors by
\begin{equation*}
\mathcal R_i( S_{ab} ) = \mathcal O_i( \nabla^a S_{ab}).
\end{equation*}
The integrability conditions for the integrands defining $\mathcal G_i(S_{ab})$ and $\mathcal R_i(S_{ab})$ ensure that the ADM momentum $\mathcal P_i$ is well defined if both $\mathcal G_i$ and $\mathcal R_i$ are, and we have
\begin{equation}
\mathcal P_i(S_{ab}) = \mathcal G_i(S_{ab}) + \mathcal R_i(S_{ab}).
\end{equation}
This decomposition depends on the specific kernel basis $k^a_{(i)}$.
In particular, since the operator $P_{g,N}$ depends on the lapse, so does the splitting. For reasons we explain below, we call $\mathcal G_i$
the \textbf{gravitational momentum}, whereas $\mathcal R_i$ is the \textbf{residual momentum}.
For the moment, it suffices to observe that if $S_{ab}$ is trace-free and satisfies the 
vacuum momentum constraint, then $\mathcal R_i(S_{ab})=0$ and hence $\mathcal P_i(S_{ab})=\mathcal G_i(S_{ab})$.

Now recall the vector field $W^a\in W^{k,p}_{(2-n)^+}$ introduced above which solves 
$(P_{g,N} W)_a = Z_a$, where $Z^a\in W^{k-2,p}_{\delta-2}$ with $1-n < \delta < 2-n$.
We claim that the ADM momentum of $1/(2N) (\ck W)_{ab}$ is well defined and is determined 
precisely in terms of the obstruction coefficients of $Z^a$:
\begin{equation*}
\mathcal P_i\left( \frac{1}{2N} (\ck W)_{ab}\right) = \mathcal O_i(Z_b).
\end{equation*}
This identity is a consequence of the following lemma, which shows that
tensors of the form $1/(2N)(\ck W)_{ab}$ cannot carry gravitational momentum 
in this gauge: that is, $\mathcal G_i(1/(2N) (\ck W)_{ab}) =0$. Hence
\begin{equation}\label{eq:pi_is_oi}
\mathcal P_i\left( \frac{1}{2N} (\ck W)_{ab} \right) = 
\underbrace{\mathcal G_i\left( \frac{1}{2N} (\ck W)_{ab} \right)}_{=0} + 
\mathcal R_i\left( \frac{1}{2N} (\ck W)_{ab} \right) = \mathcal O_i( Z_b ),
\end{equation}
as claimed.

\begin{lemma}\label{lem:W-delta-zero} 
    Suppose that $g_{ab}$ and $N$ satisfy Assumption \ref{assumption:main} with parameters $k$, $p$
    and $\tau$. If $S_{ab}\in W^{k-1,p}_{(1-n)^+}$ then $S_{ab} \nabla^a k^b_{(i)} \in L^1$ for each $i$, so 
    that $\mathcal G_i(S_{ab})$ is well defined. Moreover, if $W^a \in W^{k,p}_{(2-n)^+}$ then
    \begin{equation}
    \mathcal G_i\left(\frac{1}{2N} (\ck W)_{ab}\right) = 0.
    \end{equation}
\end{lemma}
\begin{proof}
    We first observe that 
    \be \label{eqn:nabla-k-decay}
    \nabla^a k^b_{(i)} = \nabla^a (k^b_{(i)} - e^b_{(i)}) + \nabla^a e^b_{(i)} \in W^{k-1,p}_{(-1)^-},
    \ee
    so that Lemma \ref{lem:mult-basic} implies that $S_{ab} \nabla^a k^b_{(j)} \in W^{k-1,p}_{(-n)^-} \subset L^1$ (the condition $k-1 > n/p$ ensures that such functions are continuous).  Taking $S_{ab} = \frac{1}{2N} (\ck W)_{ab}$, we have
    \begin{align*}
        \mathcal G_i(S_{ab}) & = \int_{\Reals^n} \frac{1}{2N} (\ck W)_{ab} \nabla^a k_{(i)}^b \, dV_g 
        = \int_{\Reals^n} \frac{1}{2N} (\nabla^a W^b) (\ck k_{(i)})_{ab} \, dV_g
        \\ & = \int_{\Reals^n} \nabla^a \left( \frac{1}{2N} W^{b} (\ck k_{(i)})_{ab} \right) \, dV_g - \int_{\Reals^n} W^{b}  (P_{g,N} k_{(i)})_{b} \, dV_g
        \\ & = \lim_{R \to \infty} \int_{\partial B_R} \frac{1}{2N} (\ck k_{(i)})_{ab} W^{b} \nu^a \, dA
    \end{align*}
    since $P_{g,N} k_{(i)} = 0$. By equation \eqref{eqn:nabla-k-decay} and Lemma \ref{lem:mult-basic}, however, this final integrand is contained in $W^{k-1,p}_{(1-n)^-}$, so that the integral vanishes in the limit and $\mathcal G_i(S_{ab}) = 0$.
    \\
\end{proof}

Equation \eqref{eqn:momentumints} suggests a strategy for precisely describing the asymptotics 
of the vector field $W^a$ which solves equation \eqref{eqn:pmomentum} if $Z_a$ is rapidly decreasing.  
We first identify a representative basis 
$W^a_{(i)}$, $i=1,\ldots,n$ of $O(r^{1-n})$ vectors that carry linear momentum: $\mathcal P_i(W^a_{(j)})=\delta_{ij}$.  
With these in hand, the obstructions $\mathcal O_{i}(Z_a)$ can be removed, and Proposition \ref{prop:vec-lap-properties}\eqref{part:P-solvable} then permits finding the zero 
momentum correction to the solution.  The following proposition specifies a construction 
of such a representative basis, and it yields explicit expressions for the leading order 
asymptotic expansion of the basis in terms of the Green's function\footnote{The formula appearing here is a special case of the Green's function of the classical Lam\'e operator; see Theorem 10.15 of \cite{mitrea}.} of the flat vector Laplacian $\overline P$,
\be \label{eqn:greens}
G^{ab} = - 2C_n r^{2-n}(A_n \delta^{ab} + B_n n^a n^b),
\ee
with dimensional constants 
\[
A_n = 3n-2,\qquad B_n = (n-2)^2, \qquad C_n = \frac{1}{4(n-1)(n-2)|S^{n-1}|}.
\]
We remark that we have an explicit factor of $2$ in equation \eqref{eqn:greens} due to the factor of $1/2$ in our definition of $\overline P$. For our purposes, this Green's function is most effectively leveraged in terms of the $n$ column vector fields $G_{(j)}^a$ with components $G_{(j)}^i = G^{ij}$. In the statement below, and throughout the remainder of this work, we denote the region exterior to a ball by $E_R := \{x \in \Reals^n \,:\, |x| > R\}$.

\begin{proposition}\label{prop:Wdual}
Suppose that $g_{ab}$ and $N$ satisfy Assumption \ref{assumption:main} with parameters $k$, $p$, and $\tau$. There exist vector fields $W_{(j)}^a$,  
$j=1,\ldots,n$ which satisfy the following:
\begin{enumerate}[(a)]
\item $W_{(j)}^a \in W^{k,p}_{(2-n)^+}$\label{part:Wj-decay},
\item $(P_{g,N} W_{(j)})_a$ is smooth and compactly supported, \label{part:Wj-smooth}
\item $\mathcal P_i \left( \frac{1}{2N} \ck W_{(j)} \right) = \mathcal R_i \left( \frac{1}{2N} \ck W_{(j)} \right) = \mathcal O_i \left( P_{g,N} W_{(j)} \right) = \delta_{ij}$,\label{part:W-ADM}
\item \label{part:W-decomp} for any $R>0$, there are vector fields $Y_{(j)}^a\in W^{k,p}_{(2-n+\eta)^+}$
with $\eta = \max(\tau,-1)$  such that on $E_R$
\begin{equation}\label{eqn:W-decomp}
W_{(j)}^a = G^a_{(j)} + Y_{(j)}^a.
\end{equation}
\end{enumerate}
\end{proposition}
\begin{proof}
Let $\chi(x)$ be a cutoff function that vanishes for $|x|>2$ 
and for $|x|<1/2$.  We define $\chi_R(x) := \chi(x/R)$
and let $\widetilde V_{(j)} = r^{-1-n} R \chi_R e_{(j)} \in C_c^\infty$, where $R$ sufficiently 
large is chosen below.  
Using $U_{(i)}$ as introduced in the proof of Proposition \ref{prop:kernel}, we define the quantity
\begin{align}
M_{ij}:=\int_{\Reals^n} \left<\widetilde V_{(j)}, k_{(i)}\right>_g\; dV_g
&= \int_{\Reals^n} r^{-1-n}R \chi_R \left< e_{(j)}, e_{(i)} + U_{(i)}\right>_g\; dV_g \notag \\
&= c \delta_{ij} + o(1), \notag
\end{align}
where $c$ is a constant independent of $R$.  Consequently 
$M_{ij}$ is invertible for $R$ sufficiently large: For some fixed such $R$, we take $V_{(i)} = \sum_j (M^{-1})_{ij} \widetilde V_{(j)}$  
so that 
\be \label{eqn:Wj-oi}
\int_{\Reals^n} \left< V_{(j)}, k_{(i)}\right>_g\; dV_g = \delta_{ij}.
\ee
We fix $\delta\in (2-n,0)$.
Since $V_{(j)}\in C^\infty_c \subset W^{k-2,p}_{\delta-2}$,
we can find a unique $W_{(j)}^a \in  W^{k,p}_\delta$ with 
$(P_{g,N} W_{(j)})^a = V^a_{(j)}$, and as a consequence of uniqueness the vector field $W_{(j)}$ 
is independent of $\delta\in (2-n,0)$, establishing parts \eqref{part:Wj-decay} and \eqref{part:Wj-smooth}. Equation \eqref{eqn:Wj-oi} precisely encodes that $\mathcal O_i(P_{g,N} W_{(j)}) = \delta_{ij}$, so combining this with equation \eqref{eq:pi_is_oi} establishes
part \eqref{part:W-ADM}.

To establish the decomposition \eqref{eqn:W-decomp} in part \eqref{part:W-decomp}, we fix $R>0$ and observe that
\[
(\overline P W_{(j)})_a = \underbrace{( P_{g,N} W_{(j)})_a + ((\overline P- P_{g,N}) W)_a}_{=: F^{(j)}_a}.
\]
The first term defining $F^{(j)}_a$ is smooth and compactly supported, while the second
term (using equation (\ref{eqn:pdiff}) and Lemma \ref{lem:mult-basic}) lies in $W^{k-2,p}_{(\tau-n)^+}$. Hence $F^{(j)}_a \in W^{k-2,p}_{(\tau-n)^+}$,
and Lemma \ref{lem:P-ext-exact} implies that there exists 
$X_{(j)}^a\in W^{k,p}_{(2-n+\tau)^+}$ with $(\overline P X_{(j)})_a = F^{(j)}_a$ on $E_R$.  

Since $\overline P(W_{(j)}-X_{(j)}) = 0$ on $E_R$, Lemma \ref{lem:P-multipolar}
implies that there are constants $c_{ij}$ and vector fields $U_{(j)}^a\in W^{k,p}_{(1-n)^+}$
such that
\[
W_{(j)}^a - X_{(j)}^a = \sum_{i=1}^n c_{ij} G_{(i)}^{a} + U_{(j)}^a
\]
on $E_R$. Rearranging,
\begin{equation}\label{eq:W-split-in-proof}
W_{(j)}^a = \sum_{i=1}^n c_{ij} G_{(i)}^{a} + X_{(j)}^a + U_{(j)}^a.
\end{equation}
The decay rates of $X_{(j)}^a$ and $U_{(j)}^a$ imply that $Y_{(j)}^a:= X_{(j)}^a + U_{(j)}^a \in W^{k,p}_{(2-n+\eta)^+}$
with $\eta=\max\{-1,\tau\}$, so part \eqref{part:W-decomp} is proved up to determining
the values of the constants $c_{ij}$. 

To find $c_{ij}$, we plug \eqref{eq:W-split-in-proof} into the result of part (\ref{part:W-ADM}), yielding
\[
\delta_{ij} 
= \mathcal P_i \left( \frac{1}{2N}\ck W_{(j)} \right)
= \lim_{R\to\infty}
\int_{\partial B_R} \frac{1}{2N}\left[\sum_{k=1}^n c_{kj} (\ck G_{(k)})_{ab} + (\ck Y_{(j)})_{ab}\right]n^b e^a_{(i)}\; d\overline{A}.
\]
Since $\ck Y_{(j)}$ and $(\ck - \overline \ck) G_{(k)}$ are $o(r^{1-n-\epsilon})$ for some $\epsilon > 0$, and since $N-1 = o(r^\tau)$, we deduce that
\be \label{eq:Wplug}
\delta_{ij} = \sum_{k=1}^n \frac{c_{kj}}{2} \left[ \lim_{R\to\infty} \int_{\partial B_R} (\overline \ck G_{(k)})_{ab} n^b e^a_{(i)}\; d\overline{A} \right].
\ee
Recalling equation \eqref{eqn:greens}, a computation shows that 
\be \label{eqn:LbarG}
(\overline{\ck} G_{(k)})_{ab} = \frac{4n(n-2)C_n}{r^{n-1}} \left( \delta_{ak} n_b + \delta_{bk} n_a
- \delta_{ab} n_k + (n-2) n_a n_b n_k \right),
\ee
and contracting with $n^b$ yields
\[
(\overline{\ck} G_{(k)})_{ab} n^b = 
\frac{4n(n-2)C_n}{r^{n-1}} (\delta_{ak} + (n-2) n_a n_k).
\]
Since $\int_{\partial B_R} n_a n_k\; d\overline{A} = \frac{R^{n-1}}{n}|S^{n-1}|\delta_{ak}$,
it follows that
\[
\int_{\partial B_R} (\overline{\ck} G_{(k)})_{ab} n^b e^a_{(i)}\; d\overline{A} = 
4n(n-2)|S^{n-1}| C_n  \left(1 + \frac{n-2}{n} \right) \delta_{ak} e^a_{(i)}
= 2\delta_{ik}.
\]
Recalling \eqref{eq:Wplug}, we find that $c_{ij} = \delta_{ij}$, which completes the proof of part \eqref{part:W-decomp}.
\\
\end{proof}

Armed with the momentum carrier vector fields $W_{(j)}^a$, we obtain the following precise description of the solution of equation \eqref{eqn:pmomentum} if the right-hand side is rapidly decaying:

\begin{proposition} \label{prop:fastdecay}
    Suppose that $g_{ab}$ and $N$ satisfy Assumption \ref{assumption:main} with parameters $k$, $p$, and $\tau$.
    Assume $1-n<\delta< 2-n$, and let $Z_b\in W^{k-2,p}_{\delta-2}$.
    Recall that
    \[
        \mathcal O_i(Z) = \int Z_b k^b_{(i)}\; dV_g
    \]
    and define
    \begin{equation}\label{eq:Z-projected}
    U_{b} := Z_{b} - \sum_{i=1}^n \mathcal O_i(Z) (P_{g,N}W_{(i)})_b.
    \end{equation}
    Then there exists a unique $V^a \in W^{k,p}_{\delta}$ which solves
    \begin{equation}\label{eq:PWeqB}
        (P_{g,N} V)_a = U_b.
    \end{equation}
    Hence
    \be \label{eqn:Wsolution}
        W^a =  \underbrace{\sum_{i=1}^n \mathcal O_i(Z) W_{(i)}^a}_{O(r^{2-n})} + \underbrace{V^a}_{o(r^{2-n}))}
    \ee
    is the unique solution in $W^{k,p}_{(2-n)^+}$ of equation \eqref{eqn:pmomentum}.
    Moreover, the ADM momentum of $1/(2N)(\ck W)_{ab}$ satisfies
    \[
    \mathcal P_i\left( \frac{1}{2N}(\ck W) \right) = \mathcal O_i(Z), \quad i=1,\ldots,n.
    \]
\end{proposition}
\begin{proof}
    Since each $(P_{g,N} W_{(j)})_a$ is smooth and compactly supported,
    $U_b$ defined by equation \eqref{eq:Z-projected} has the same regularity and decay properties 
    as $Z_b$. That is, $U_b \in W^{k-2,p}_{\delta-2}$. Moreover,
    since $\mathcal O_i(P_{g,N}W_{(j)}) = \delta_{ij}$ we find that
    \[
    \mathcal O_i(U) = \mathcal O_i(Z) - \sum_{j=1}^n \mathcal O_j(Z) \mathcal O_i( P_{g,N} W_{(j)}) = 0.
    \]
    Since the obstruction coefficients vanish, 
    Proposition \ref{prop:vec-lap-properties}\eqref{part:P-solvable} 
    provides a unique solution $V^a \in W^{k,p}_{\delta}$ of \eqref{eq:PWeqB}. 
    Hence $W^a$ defined by \eqref{eqn:Wsolution} is the unique solution
    in $W^{k,p}_{(2-n)^+}$ of equation \eqref{eqn:pmomentum} provided 
    by Proposition \ref{prop:vec-lap-properties}.  Finally,
    using Lemma \ref{lem:W-delta-zero} we compute the ADM momentum as follows:
    \[
    \mathcal P_i\left(\frac{1}{2N} \ck W\right) = \mathcal R_i\left(\frac{1}{2N} \ck W\right) + \mathcal G_i\left(\frac{1}{2N} \ck W\right) = \mathcal O_i(P_{g,N} W)  = \mathcal O_i(Z).
    \]    
\end{proof}

Combining Proposition \ref{prop:fastdecay} and Proposition \ref{prop:Wdual}\eqref{part:W-decomp} 
we find that if $Z_a$ has fast decay, $o(r^{-n})$, then the solution of
equation \eqref{eqn:pmomentum} has a leading order $O(r^{2-n})$ term that is a linear combination of 
the columns of the Green's function of the Euclidean vector Laplacian, with a correction decaying
at a faster rate depending on both $\tau$ and $\delta$.  We now investigate the borderline case
such that $Z_a$ decays at the threshold rate $r^{-n}$. The situation is more delicate because the vector Laplacian is no longer Fredholm as a map on the full weighted space $W^{k,p}_{2-n}$.
To compensate for this, we confine our attention to the case such that $Z_a$ admits the expansion $Z_a=r^{-n} F_a + H_a$ on some exterior region $E_R$, where $F_a$ is homogeneous of degree zero and where $H_a$ lies in a weighted space with decay faster than $r^{-n}$. 
In this scenario, we seek a solution of the form 
$W^a= r^{2-n} U^a + V^a$, 
where $U^a$ is homogeneous of degree zero and where $V^a$ decays faster than $r^{2-n}$.  
The following result concerning the Euclidean vector Laplacian is the tool needed for the solvability for the leading order homogeneous terms:
\begin{proposition} \label{prop:flat-P-threshold}
    Set $k\ge 2$ and $1<p<\infty$.
    Suppose that the covector field $Z_a\in W^{k-2,p}_{\mathrm{loc}}(\Reals^n\setminus\{0\})$ has the form $Z_a =r^{-n} F_a$, where $F_a \in W^{k-2,p}(S^{n-1})$ is extended homogeneously with
    degree zero.  Then there exists a vector field $W^a \in W^{k,p}_{\mathrm{loc}} (\Reals^n\setminus\{0\})$ of the form $W^a = r^{2-n} U^a$, where $U^a \in W^{k,p}(S^{n-1})$ is extended homogeneously with degree zero, solving $(\overline P W)_a = Z_a$    
    if and only if
    \be\label{eq:homog-solve}
    \int_{S^{n-1}} F(e_{(i)}) \; d\overline A = 0 \quad \text{for} \quad i=1,\ldots,n.
    \ee
    The solution $W^a$ is unique up to a linear combination of the vector fields $G^a_{(j)}$.
\end{proposition}
\begin{proof}
Consider a smooth vector field $W^a$ on $\Reals^{n}\setminus\{0\}$ of the form $r^\alpha( V^a + v n^a)$ 
where $\alpha\in\Reals$, and where 
$V^a$ is a vector field which is contained in $W^{k,p}(S^{n-1})$ and is tangential to $S^{n-1}$ and $v$ is a function which is contained in $W^{k,p}(S^{n-1})$, both extended to $\Reals^{n}\setminus\{0\}$ by zero-homogeneity. An extended computation shows that 
$(\overline P W)^a= \frac{1}{2}r^{\alpha-2}(F^a + fn^a)$, where $F^a$ and $f$ are a tangential 
vector field and a function respectively contained in $W^{k-2,p}(S^{n-1})$ with
\[
\begin{bmatrix} F^a \\ f\end{bmatrix} = 
\begin{bmatrix} \slashed\Delta + \left(1-\frac{2}{n}\right) \slashed{\nabla}\slashed{\Div} + a_\alpha & b_\alpha \slashed \nabla \\
    c_\alpha \slashed \Div & \slashed \Delta + d_\alpha\end{bmatrix}
\begin{bmatrix} V^a \\ v\end{bmatrix}.
\]
Here, the slashed operators are intrinsic to the round $S^{n-1}$ and the coefficients are
\begin{align*}
    a_\alpha &= \alpha(\alpha+n-2)-1\\
    b_\alpha &= n + (\alpha-1)\left(1-\frac{2}{n}\right)\\
    c_\alpha &= \alpha -3 - \frac{2}n(\alpha-1)\\
    d_\alpha &= 2\frac{(\alpha-1)(n-1)}{n}(\alpha+n-1)
\end{align*}
The operators $L_\alpha$ sending $(V^a, v) \mapsto (F^a, f)$ as above are elliptic and Fredholm as maps $W^{k,p}\to W^{k-2,p}$ on the bundle over $S^{n-1}$ with sections being pairs of tangential vector fields and functions mapping $S^{n-1} \to \Reals$. 
These operators are homotopic via Fredholm maps to $\mathrm{diag}(\slashed \Delta, \slashed \Delta)$,
which is formally self adjoint and hence has index zero.  The maps $L_\alpha$ then
also have index zero.  They are not self adjoint in general, but since
\[
a_0 = a_{2-n};\quad b_0 = -c_{2-n};\quad c_0 = -b_{2-n};\quad d_0 = d_{2-n},
\]
and since $(\slashed \Div)^* = -\slashed \nabla$, it follows that the formal adjoint
of $L_{2-n}$ is $L_0$.  Hence, for any choice of $(F^a,f)\in W^{k-2,p}(S^{n-1})$ the
equation $L_{2-n}(V^a, v) = (F^a, f)$ is solvable for $(V^a,v)\in W^{k,p}(S^{n-1})
$ if and only if $(F^a,f)$ is $L^2$ orthogonal to the kernel of $L_0$.

To identify this kernel, it is helpful to return to Cartesian coordinates, wherein the action of $\overline P$ takes the form
\[
2(\overline P X)_a = \overline \Delta(X_a) + \left(1-\frac{2}{n}\right) \overline \nabla_a(\overline \nabla_b X^b),
\]
where $\overline \Delta$ applies {\it component-wise} on the right hand side.
We are then seeking vector fields $X^a$ that are homogeneous of degree 0 and are contained in the 
kernel of $\overline P$. We claim that such a vector field must be divergence-free, satisfying $\overline{\nabla}_b X^b = 0$. Assuming this is true,
then
\[
0 = 2(\overline P X)_a = \overline\Delta(X_a).
\]
Hence the components of $X^a$ are harmonic
and homogeneous of degree 0, and it follows that they are constants. Conversely, constant vector fields 
are clearly contained in the kernel, so we conclude that the kernel of $L_0$ precisely consists of the 
constant vector fields on $\Reals^n$, restricted to the sphere and decomposed into tangential 
and normal components. The solvabilty conditions \eqref{eq:homog-solve} then follow.

To establish the claim that $\overline \nabla_b X^b = 0$, we observe that
\[
0 = \overline \nabla^a (\overline P X)_a =  \left(1-\frac1n\right) \overline\Delta(\overline\nabla_b X^b).
\]
Since $1-(1/n)\neq 0$, we conclude that $\overline\nabla_b X^b$ is harmonic and homogeneous of degree $-1$.
If $n>3$ there are no such nonzero functions since $2-n < -1 < 0$, which establishes the claim in this case.  Now suppose that $n=3$, in which case $\overline\nabla_b X^b = c/r$ for some constant $c$.  We need to 
show $c=0$.  

Since $(\overline P X)_a=0$ and since $\overline\nabla_b X^b = c/r$ we have
\[
\overline \Delta (X_a) = -\left(1-\frac{2}{3}\right) \overline \nabla_a (c/r) = \frac{c}{3r^2}n_a.
\]
Using the fact that the first spherical harmonics on $S^2$ (i.e., the 
restrictions of the components of $n^a$ to the sphere) all have the eigenvalue $-2$ with respect to the Laplacian, we find that
$
X^a = C^a -\frac{c}{6} n^a
$
for some constant vector field $C^a$.  Taking a divergence, we compute that $\overline\nabla_b X^b = -c/(3r)$.
Recalling that we began with $\overline\nabla_b X^b = c/r$, we conclude that $c=0$, as required.

To finish the proof, it remains to identify the kernel of $L_{2-n}$. 
The $n$ vector fields $G^{a}_{(j)}$ from the Greens function of $\overline P$ are homogeneous of degree $2-n$,
and each provides an element of this kernel. Since $L_{2-n}$ has index zero 
and a cokernel of dimension $n$, its kernel has dimension $n$, so these provide a complete basis.
\\
\end{proof}

With this result for the Euclidean operator $\overline P$ in hand, we can now solve the broader problem for $P_{g,N}$, supposing that our leading order source term $F_a$ satisfies the solvability condition:

\begin{proposition} 
\label{prop:W-threshold-decomp}
    Suppose that $g_{ab}$ and $N$ satisfy Assumption \ref{assumption:main} with parameters $k$, $p$, and $\tau$.
    Suppose that the covector field $Z_a\in W^{k-2,p}_{(-n)^+}$ admits the decomposition
    \[
    Z_a = r^{-n} F_a + H_a
    \]
    on $E_R$ for some $R>0$, where $H_a\in W^{k-2,p}_{\delta-2}$ for some $1-n < \delta < 2 - n$, and where $F_a \in W^{k-2,p}(S^{n-1})$ is extended homogeneously with
    degree zero. If 
    \be\label{eq:homog-obstruct}
    \int_{S^{n-1}} F(e_{(i)})\; d \overline A = 0 \quad \text{for} \quad i=1,\ldots,n,
    \ee
    then the unique solution $W^a \in W^{k,p}_{(2-n)^+}$ of $(P_{g,N} W)_a = Z_a$ admits the decomposition
    \be \label{eqn:W-threshold-decomp}
    W^a = r^{2-n} U^a + V^a
    \ee
    on $E_R$, where $V^a \in W^{k,p}_{\delta} \cup W^{k,p}_{(2-n+\tau)^+}$ and $U_a \in W^{k,p}(S^{n-1})$ is extended homogeneously with degree zero.
\end{proposition}
\begin{proof}
Proposition \ref{prop:flat-P-threshold} ensures that we may find $\overline W^a \in W^{k,p}_{\mathrm{loc}}(\Reals^n \setminus \{0\})$ of the form $\overline W^a = r^{2-n} \overline U^a$, with $\overline U^a\in W^{k,p}(S^{n-1})$ extended homogeneously with degree zero, and with $\overline W^a$ satisfying $( \overline P \,\overline W)_a = r^{-n} F_a$. This determines $\overline W^a$ up to a linear combination of the $G^a_{(j)}$ vector fields; we fix one particular choice of $\overline W^a$ here.

Let $\chi$ be a cutoff function that equals $0$ on $B_{R/2}$  and $1$ on $E_R$. Then $\overline V^a := W^a - \chi \overline W^a \in W^{k,p}_{(2-n)^+}$ satisfies 
\begin{align*}
(P_{g,N} \overline V)_a & = (P_{g,N}W)_a - (\overline P\, \chi \overline W)_a - ((P_{g,N}- \overline P)\chi \overline W)_a 
\\ & = H_a - ((P_{g,N}- \overline P) \chi \overline W)_a
\end{align*}
on $E_R$. Since $\chi \overline W^a \in W^{k,p}_{(2-n)^+}$, the latter term is contained in $W^{k-2,p}_{(\tau-n)^+}$, so that $(P_{g,N} \overline V)_a \in W^{k-2,p}_{\delta-2} \cup W^{k-2,p}_{(\tau-n)^+}$. As a consequence of Propositions \ref{prop:fastdecay} and \ref{prop:Wdual}\eqref{part:W-decomp}, $\overline V^a$ admits the decomposition
\[
\overline V^a = \sum_{j=1}^n \mathcal O_j(P_{g,N} \overline V) G^a_{(j)} + V^a
\]
on $E_R$, where $V^a \in W^{k,p}_{\delta} \cup W^{k,p}_{(2-n+\tau)^+}$. Since we have on $E_R$ that
\[
W^a = \overline W^a + \overline V^a = r^{2-n} \overline U^a + \sum_{j=1}^n \mathcal O_j(P_{g,N} \overline V) G^a_{(j)} + V^a,
\]
this establishes \eqref{eqn:W-threshold-decomp} with $U^a = \overline U^a + \sum_{j=1}^n \mathcal O_j(P_{g,N} \overline V) r^{n-2} G^a_{(j)}$.
\\
\end{proof}

We now show that the obstruction criteria \eqref{eq:homog-obstruct} of Proposition \ref{prop:W-threshold-decomp} are automatically satisfied for our primary use case where the source term is a divergence.

\begin{lemma}
\label{lem:div-obstr-vanish}
Suppose $S_{ab}\in W^{k-1,p}_{\mathrm{loc}}(\Reals^n\setminus\{0\})$
is homogeneous of degree $1-n$.  Then 
\be\label{eq:div-S-form}
\overline{\nabla}^a S_{ab} = r^{-n} F_b
\ee
for some $F_b \in W^{k-2,p}(S^{n-1})$ extended homogeneously with degree zero and
satisfying
\be\label{eq:div-obstr-vanish}
\int_{S^{n-1}} F(e_{(i)}) \; d\overline A = 0\quad \text{for} \quad i=1,\ldots, n.
\ee
\end{lemma}
\begin{proof}
The form for $\overline \nabla^a S_{ab}$ given by equation \eqref{eq:div-S-form} is
a fairly direct consequence of homogeneity, so it remains to establish 
the vanishing of the obstruction coefficients in equation \eqref{eq:div-obstr-vanish} .  

Consider $r_1, r_2\in \Reals$ with $0<r_1 < r_2 $ and let 
$A$ be the annulus bounded by $r_1 < r < r_2$.  
Then
\[
\int_A \overline \nabla^a(S_{ab}) e^b_{(i)}\; dV_{\overline g} = \int_{r_1}^{r_2} r^{-1} \int_{S^{n-1}} F_a e^a_{(i)}\; d\overline A\; dr = \log(r_2/r_1) \int_{S^{n-1}} F_a e^a_{(i)}\; d \overline A.
\]
On the other hand, the divergence theorem implies that
\begin{align*}
\int_A \overline \nabla^a(S_{ab}) e^b_{(i)}\; dV_{\overline g} & = 
\int_A \overline \nabla^a(S_{ab} e^b_{(i)})\; dV_{\overline g} = 
\int_{\partial B_{r_2}} S_{ab} e^b_{(i)} n^a \; dA_{\overline g} 
- \int_{\partial B_{r_1}} S_{ab} e^b_{(i)} n^a \; dA_{\overline g} 
\\ & =
\int_{S^{n-1}} B_{ab} e^b_{(i)} n^a \; d\overline A
-
\int_{S^{n-1}} B_{ab} e^b_{(i)} n^a \; d\overline A
=0.
\end{align*}
Equating the final expressions from the two computations above, we find that
\[
\log(r_2/r_1) \int_{S^{n-1}} F_a e^a_{(i)}\; d\overline A = 0,
\]
which establishes the result.
\\
\end{proof}

Compiling the results in this section, we have the following theorem summarizing the prescription of the asymptotics of the unphysical second fundamental form $K_{ab}$ in the conformal method:
\begin{theorem}
\label{thm:momentum-summary}
    Suppose that $g_{ab}$ and $N$ satisfy Assumption \ref{assumption:main} with parameters $k$, $p$, and $\tau$.
    We choose $A_{ab}\in W^{k-1,p}_{\delta-1}$ for some $\delta<0$ to be symmetric and trace-free,
    and we define $Z_b = \nabla^a \left(\frac{1}{2N} A_{ab}\right)$. Given $W^a$ solving $(P_{g,N} W)_a = Z_a$, we set $K_{ab} = 1/(2N)(A_{ab}-(\ck W)_{ab})$ and $S_{ab} = 1/(2N) A_{ab}$. The following holds:

    \begin{enumerate}[(a)]
    \item If $2-n<\delta<0$, then there exists a unique $W^a\in W^{k,p}_{\delta}$ 
    which solves  $(P_{g,N} W)_a = Z_a$.
    \label{part:K-slow-falloff}

    \item If $1-n<\delta<2-n$, then the unique $W^a \in W^{k,p}_{(2-n)^+}$ which solves $(P_{g,N} W)_a = Z_a$ may be expanded in terms of the momentum carriers,
    \[
    W^a = \sum_{i=1}^n \mathcal O_i(Z) W^{a}_{(i)} + V^a,
    \]
    with the remainder $V^a \in W^{k,p}_{\delta}$.  Moreover, $\mathcal P_i(K) = \mathcal G_i(S) = \mathcal P_i(S) - \mathcal R_i(S) = -\mathcal O_i(Z)$.
    \label{part:K-fast-falloff}

    \item In the borderline case such that $A_{ab}\in W^{k-1,p}_{(1-n)^+}$ has the form
    \be \label{eqn:A-threshold-decomp}
    A_{ab} = r^{1-n} B_{ab} + C_{ab}
    \ee
    on $E_R$ for some $R>0$, where $B_{ab} \in W^{k-1,p}_{\mathrm{loc}}(\Reals^n\setminus\{0\})$ is homogeneous of degree zero and $C_{ab}\in W^{k-1,p}_{\delta-1}$ 
    for some $1-n<\delta<2-n$ (with $B_{ab}$ and $C_{ab}$ both traceless and symmetric), the unique $W^a \in W^{k,p}_{(2-n)^+}$ which solves $(P_{g,N} W)_a = Z_a$ takes the form
    \be \label{eqn:W-threshold-decomp2}
    W^a = r^{2-n} U^a + V^a
    \ee
    on $E_R$, where $U^a \in W^{k,p}_{\mathrm{loc}}(\Reals^n\setminus\{0\})$ is homogeneous of degree zero and $V^a\in W^{k,p}_{\delta} \cup W^{k,p}_{(2-n+\tau)^+}$.
    Moreover, $\mathcal P_i(K) = \mathcal G_i(S) = \mathcal P_i(S) - \mathcal R_i(S)$.
    \label{part:K-threshold-falloff}
    \end{enumerate}
\end{theorem}
\begin{proof}
Part \eqref{part:K-slow-falloff} follows immediately from Proposition \ref{prop:vec-lap-properties}\eqref{part:iso}. 

Recalling that $\mathcal G_i(S) = \mathcal P_i(S) - \mathcal R_i(S)$ holds definitionally, the momentum identifications $\mathcal P_i(K) = \mathcal G_i(S) = \mathcal P_i(S) - \mathcal R_i(S)$ in both parts \eqref{part:K-fast-falloff} and \eqref{part:K-threshold-falloff} follow from expanding $\mathcal P_i(K)$ into $\mathcal G_i(K)$ and $\mathcal R_i(K)$ and using that $\mathcal R_i(K) = 0$ (which holds because $\nabla^a K_{ab} = 0$) and that $\mathcal G_i \left( \frac{1}{2N} \ck W \right) = 0$ (from Lemma \ref{lem:W-delta-zero}). Explicitly, we have the following:
$$
\mathcal P_i(K) = \mathcal G_i(K) + \mathcal R_i(K) = \mathcal G_i(S) - \mathcal G_i \left( \frac{1}{2N} \ck W \right) = \mathcal G_i(S),
$$
The decomposition of $W^a$ in part \eqref{part:K-fast-falloff} follows immediately from Proposition \ref{prop:fastdecay}. The further momentum identification $\mathcal P_i(K) = - \mathcal O_i(Z)$ in \eqref{part:K-fast-falloff} follows from the fact that the fast decay of $S_{ab}$ results in $\mathcal P_i(S) = 0$, as well as from the definitional relation $\mathcal R_i\left( S \right) = \mathcal O_i(Z)$.

The decomposition in part \eqref{part:K-threshold-falloff} follows from Proposition \ref{prop:W-threshold-decomp} since Lemma \ref{lem:div-obstr-vanish} ensures that $\overline \nabla^a(\frac{1}{2N} r^{1-n} B_{ab}) = r^{-n} F_b$ for some $F_b \in W^{k-2,p}(S^{n-1})$ which satisfies \eqref{eq:div-obstr-vanish} and which is extended homogeneously with degree zero, so that
on $E_R$ we may write
\begin{align*}
Z_b = \nabla^a\left( \frac{1}{2N} A_{ab} \right)  = \overline \nabla^a\left( \frac{1}{2N} r^{1-n} B_{ab} \right) + (\nabla^a - \overline \nabla^a)\left( \frac{1}{2N} r^{1-n} B_{ab} \right) + \nabla^a\left(\frac{1}{2N} C_{ab} \right) = r^{-n} F_b + H_b,
\end{align*}
where we have (choosing $\chi$ to be a cutoff function that equals $0$ on $B_{R/2}$  and $1$ on $E_R$)
\[
H_b := \chi (\nabla^a - \overline \nabla^a)\left( \frac{1}{2N} r^{1-n} B_{ab} \right) + \nabla^a\left(\frac{1}{2N} C_{ab} \right) \in W^{k-2,p}_{\delta-2} \cup W^{k-2,p}_{(\tau-n)^+}.
\]
\end{proof}

This summary indicates the extent of control one has over the structure of the unphysical second fundamental form $K_{ab}$ with regards to its ADM momentum and decay properties.
In particular, if one has sufficient decay to make sense of each of $\mathcal G_i(S)$ and $\mathcal R_i(S)$, the ADM momentum components are $\mathcal P_i(K) = \mathcal G_i(S) = \mathcal P_i(S) - \mathcal R_i(S)$.
For any conformal method seed data $g_{ab}$ and $N$ which satisfy Assumption \ref{assumption:main} and $A_{ab}$ with at least $O(r^{1-n})$ decay, $\mathcal G_i(A_{ab}/2N)$ therefore directly measures the ADM momentum of the eventual initial data set, and it follows that this momentum can be readily prescribed a priori. Moreover, the quantity $\mathcal R_i(A_{ab}/2N)$ specifies the momentum which must be removed to satisfy the momentum constraint. These results motivate our use of the monikers ``gravitational momentum" for $\mathcal G_i$ and ``residual momentum" for $\mathcal R_i$. 

In all cases, one can control the decay rate of $K_{ab}$ (beyond any $r^{1-n}$ terms that arise) by appropriately choosing the decay rates of one's seed data. 
To illustrate, we now discuss the precise decay structure of $W^a$ and $K_{ab}$ which we obtain in solving \eqref{eqn:pmomentum} in the fast-decay setting of $A_{ab}\in W^{k-1,p}_{\delta-1}$ with $1-n < \delta < 2-n$, corresponding to part \eqref{part:K-fast-falloff} above. 
Abbreviating $p_i := \mathcal G_i(S) = - \mathcal O_i(Z)$, we find that $W^a$ contains an $O(r^{2-n})$ piece of the form  $\sum_j -p_j W_{(j)}^a$, and the remainder $V^a$ beyond this has the desired $o(r^\delta)$ decay inherited from $A_{ab} \in W^{k-1,p}_{\delta-1}$. 
In light of Proposition \ref{prop:Wdual}\eqref{part:W-decomp}, however, one can be more explicit in describing the asymptotic form of $W^a$: For every $\epsilon > 0$ and any $\ell \in \Nats$ such that $k-1 > \ell+n/p$, one has
\[
W^a = \sum_{j=1}^n -p_j G_{(j)}^a + o_{\ell+1}(r^{\max(\delta, \, 2-n+\tau+\epsilon)}),
\]
and the corresponding asymptotics of $K_{ab}$ are the following (using \eqref{eqn:LbarG}):
\begin{align}
K_{ab} & = \frac{1}{2}\sum_{j=1}^n p_j \cdot  (\overline \ck G_{(j)})_{ab} + o_\ell (r^{\max(\delta-1, \, 1-n+\tau+\epsilon)}) \notag \\
& =  \frac{2n(n-2) C_n}{r^{n-1}} \left[ p_a n_b + p_b n_a + ((n-2) n_a n_b - \delta_{ab}) p_j n^j \right] + o_\ell(r^{\max\{\delta-1, \, 1-n+\tau+\epsilon\}}). \label{eqn:K-fast-decay}
\end{align}
In particular, for $n=3$ one has
\[
K_{ab} =  \frac{3}{16 \pi r^{2}}  \left[ p_a n_b + p_b n_a + (n_a n_b - \delta_{ab}) p_j n^j \right] + o_\ell(r^{\max \{ \delta-1, \, -2+\tau+\epsilon \} }).
\]
We note that the leading-order contribution here is the linear-momentum Bowen-York tensor with the momentum vector $p_a$. These are Euclidean transverse traceless tensors which have long been used as tools in the construction of black hole initial data sets with prescribed momentum in the conformally flat setting \cite{bowen_york}. 

In the case of fast decay, equation \eqref{eqn:K-fast-decay} illustrates that we may anticipate the complete structure of the leading $O(r^{1-n})$ term, and that we may dictate the decay rate of $K_{ab}$ (and hence of $\widetilde K_{ab}$) beyond this term through the decay of the chosen seed data $g_{ab}$ and $N$ (specifying $\tau$) and $A_{ab}$ (specifying $\delta$). 
One obtains as many derivatives of control in the error term as desired by choosing the seed data to allow $k$ sufficiently large. 
Although it is not typically of physical interest, we remark that one could obtain similar results down to $-n < \delta < 1-n$, or for an arbitrarily negative value of $\delta$, by expanding Proposition \ref{prop:kernel} to characterize the kernel of $P_{g,N}$ on higher-weighted spaces and expanding Proposition \ref{prop:Wdual} to introduce the faster-decaying carrier vector fields dual to these additional kernel elements.

While the leading $O(r^{1-n})$ term of $K_{ab}$ must be a linear-momentum Bowen-York tensor in the case of fast decay discussed above, there is more freedom if $A_{ab}$ is itself chosen with the threshold $O(r^{1-n})$ decay, as in part \eqref{part:K-threshold-falloff} of Theorem \ref{thm:momentum-summary}. 
A question then remains regarding what kinds of leading terms $r^{1-n} D_{ab}$ ($D_{ab}$ is analogous to the quantity in brackets in equation \eqref{eqn:K-fast-decay}) are possible for $K_{ab}$ in this case. 
The answer is that one can choose any Euclidean transverse traceless tensor of this form: the possible homogeneous of degree zero tensors $D_{ab} \in W^{k-1,p}_{\mathrm{loc}}(\Reals^n \setminus \{0\})$ are precisely those which are symmetric and traceless and satisfy $\overline \nabla^a(r^{1-n}D_{ab}) = 0$.
This form is necessary because isolating the leading order $r^{1-n} D_{ab}$ term in $K_{ab} = \frac{1}{2N}(A_{ab}-(\ck W)_{ab})$, with $A_{ab}$ and $W^a$ as in \eqref{eqn:A-threshold-decomp} and \eqref{eqn:W-threshold-decomp2} on some $E_R$, yields that the symmetry and tracelessness of $D_{ab}$ are consequences of these properties of $K_{ab}$ overall, and the leading order term in $\nabla^a K_{ab} = 0$ reads $\overline \nabla^a(r^{1-n}D_{ab}) = 0$.

Conversely, beginning with any chosen Euclidean transverse traceless tensor of them form $r^{1-n} D_{ab}$, if one sets 
\[
A_{ab} = 2N\chi r^{1-n} D_{ab} \in W^{k-1,p}_{(1-n)^+}
\]
with $\chi$ a cutoff function equal to $0$ on $B_{R/2}$ and $1$ on some $E_R$, then on $E_R$ one has
\[
Z_b := \nabla^a\left( \frac{1}{2N} A_{ab} \right) = \nabla^a\left(r^{1-n} D_{ab} \right) = (\nabla^a - \overline \nabla^a)\left(r^{1-n} D_{ab} \right) = o(r^{\tau-n}),.
\]
This result implies that $Z_b \in W^{k-2,p}_{(-n)^-}$, and Proposition \ref{prop:fastdecay} therefore yields the decomposition \eqref{eqn:Wsolution} for the unique solution $W^a \in W^{k,p}_{(2-n)^+}$ to $(P_{g,N} W)_a = Z_a$. 
The leading order $O(r^{1-n})$ terms in $K_{ab} := \frac{1}{2N}(A_{ab}-(\ck W)_{ab})$ then include both $r^{1-n}D_{ab}$ and the leading Bowen-York tensor in equation \eqref{eqn:K-fast-decay}, with $ p_j = -\mathcal R_j(r^{1-n} D_{ab})$. 
This extraneous Bowen-York term can be removed by perturbing $A_{ab}$ at lower order (specifically, by subtracting some $B_{ab} \in W^{k-1,p}_{(1-n)^-}$ which satisfies $\mathcal R_j( \frac{1}{2N}B_{ab}) = \mathcal R_j(r^{1-n} D_{ab})$), resulting in $K_{ab}$ having only $r^{1-n} D_{ab}$ at leading order, as desired.

We remark that the problem of characterizing and constructing Euclidean transverse traceless tensors is well-studied in the literature \cite{beig1997tt,tafel2018all}, so numerous such objects can be readily constructed. The procedure outlined in part \eqref{part:K-threshold-falloff} of Theorem \ref{thm:momentum-summary} and the above discussion then allows one to build any example into a working $K_{ab}$ as the leading order $O(r^{1-n})$ term if desired.

\section{The Hamiltonian Constraint}
\label{sec:hamiltonian}

Assured that we may construct a continuous $K_{ab}$ with any desired asymptotics, we suppose now that we have done this and turn to the vacuum Hamiltonian constraint, which in the conformal method takes the form of the Lichnerowicz equation (\ref{eqn:lichnerowicz}),
\be \label{eqn:lichnerowicz2}
-(q_n+2)\Delta_g \varphi + R(g) \varphi = \varphi^{-q_n-1} |K|_g^2
\ee
(we recall that $q_n := \frac{2n}{n-2}$), to be solved for the conformal factor $\varphi > 0$ with the imposed boundary condition $\lim_{|x| \to \infty} \varphi = 1$. Despite the nonlinearity of this equation, the issue of determining existence and uniqueness of solutions has been thoroughly resolved in the literature \cite{bruhat-isenberg,maxwell_solutions_2005}. One first observes that, in addition to Assumption \ref{assumption:main}, it is necessary to require that $g$ is {\it Yamabe positive}, which means that the conformal invariant
\be
Y_g := \inf_{\substack{f\in C^\infty_c\\ f\not\equiv 0}} \frac{\int_{\Reals^n} \left[ (q_n+2) |\nabla f|^2 + R(g) f^2 \right] \; dV_g}{||f||_{L^{q_n}}^2}
\ee
is strictly positive. To see that this is a necessary condition, we suppose that there exists a positive solution $\varphi$ to equation \eqref{eqn:lichnerowicz2} with $\varphi-1 \in W^{k,p}_\delta$ for some $\delta < 0$. Under the conformal change specified by $\varphi$, one has $Y_{\tilde g} > 0$ because $R(\tilde g) = |K|_{\tilde g}^2$ is positive and continuous, and because the Sobolev embedding ensures that
\[ \|f\|_{L^{q_n}}^2 \lesssim  \|\nabla f\|_{L^2}^2 + \| R(\tilde g)^{1/2} f\|_{L^2}^2 .\]
The notation ``$\lesssim$" here indicates that one has a standard inequality, with ``$\leq$", up to an overall positive multiplicative constant independent of the quantities of interest (in this case, independent of $f \in C_c^\infty$).
Conformal invariance now yields $Y_g > 0$. 

It turns out that $Y_g > 0$ is also a sufficient condition, as indicated by the following existence theorem for solutions of the Lichnerowicz equation for asymptotically Euclidean seed data. This theorem is implied by Theorem 4.4 of \cite{maxwell_solutions_2005}; for uniqueness, see the discussion in Section VIII of \cite{bruhat-isenberg}.

\begin{theorem} \label{thm:lichnerowicz-existence-slow}
Suppose that $g_{ab}$ satisfies Assumption \ref{assumption:main} with parameters $k$, $p$, and $\tau$ with $2-n < \tau < 0$, and that $K_{ab} \in W^{k-1,p}_{\tau-1}$. If $Y_g > 0$, then there exists a unique positive conformal factor $\varphi$ which solves the Lichnerowicz equation \eqref{eqn:lichnerowicz2} with $\varphi - 1 \in W^{k,p}_\tau$.
\end{theorem}
\qed

It is useful to also state the implication of this result in the case of faster decay:

\begin{corollary} \label{cor:lichnerowicz-existence-fast}
Suppose that $g_{ab}$ is asymptotically Euclidean of class $W^{k,p}_{(2-n)^+}$ with $k > 1+n/p$, and that $K_{ab} \in W^{k-1,p}_{(1-n)^+}$. If $Y_g > 0$, then there exists a unique positive conformal factor $\varphi$ which solves the Lichnerowicz equation \eqref{eqn:lichnerowicz2} with $\varphi - 1 \in W^{k,p}_{(2-n)^+}$.
\end{corollary}
\qed

As in Section \ref{sec:momentum}, understanding the asymptotics of these solutions requires a finer analysis of the next decay range $1-n< \delta < 2-n$. Since we are now dealing with the simpler scalar Laplacian rather than the vector Laplacian, much of the requisite analysis is a simplified version of that developed in Section \ref{sec:momentum} and is well understood in the literature. To make the relation between the analyses of these two Laplacians explicit, we present this theory in a manner closely analogous to the discussion in Section \ref{sec:momentum}. To most naturally incorporate the boundary condition $\varphi \to 1$, we write $\varphi = 1 + u$ and recast \eqref{eqn:lichnerowicz2} as follows:
\be \label{eqn:lichnerowicz-u}
-(q_n+2) \Delta_g u + (1+u)R(g) = (1+u)^{-q_n-1} |K|_g^2
\ee
This is the equation we analyze. We begin our analysis by noting the close connection between $\Delta_g$ and the Euclidean Laplacian $\overline \Delta$ induced by the following equation:
\be \label{eqn:lap-diff}
\Delta_g u = \overline \Delta u + \sum_{|\alpha| = 1}^2 h^\alpha \partial_\alpha u,
\ee
with each $h^{\alpha} \in W^{k-2+|\alpha|,p}_{\tau-2+|\alpha|}$.
Since the following statement is well known, and since the details are much the same as those of Proposition \ref{prop:vec-lap-properties}, we omit the proof.

\begin{proposition}\label{prop:scalar-lap-fredholm}
    Suppose that $g_{ab}$ satisfies Assumption \ref{assumption:main} with parameters $k$, $p$, and $\tau$, 
    and that $V \in W^{k-2,p}_{\tau-2}$. If $\delta$ is non-exceptional, then
    $
        L = -\Delta_g + V : W^{k,p}_\delta \to W^{k-2,p}_{\delta-2}
    $
    is continuous and Fredholm, with its index equal to that of the Euclidean Laplacian acting between the same
    spaces. If $V \geq 0$, then $L$ is an isomorphism for $2-n < \delta < 0$.
\end{proposition}
\qed

The following proposition and corollary are a variation of Proposition 1 in \cite{bieri_brill_2025} in that they extract the mass term from $\varphi$ to control the remaining decay. Following the thread of Section  \ref{sec:momentum}, we reformulate this as relying on the existence of a ``mass carrier" function-- compare this to Proposition \ref{prop:Wdual} (the constant function $1$ now plays the role analogous to that of the kernel vector fields $k_{(j)}^a$).

\begin{proposition}\label{prop:u0}
Suppose that $g_{ab}$ satisfies Assumption \ref{assumption:main} with parameters $k$, $p$, and $\tau$.
There exists a function $u_0\in W^{k,p}_{(2-n)^+}$ which satisfies the following:
\begin{enumerate}[(a)]
\item $-\Delta_g u_0$ is smooth, nonnegative, and compactly supported,
\item $\int_{\Reals^n} (-\Delta_g u_0) \; dV_g = 1$,
\item \label{part:u0exp} for each $R>0$, there is a function $v\in W^{k,p}_{(2-n+\eta)^+}$ with $\eta=\max(\tau,-1)$
such that
\begin{equation}\label{eqn:u0expansion}
u_0 = \frac{r^{2-n}}{(n-2)|S^{n-1}|} + v
\end{equation}
on the exterior region $E_R$.
\end{enumerate}
\end{proposition}
\begin{proof}
We choose $\mu$ to be a smooth, non-negative, compactly supported function with $\int_{\Reals^n} \mu \; dV_g=1$.
Hence $\mu\in W^{k-2,p}_{\delta}$ for every $\delta$, so Proposition \ref{prop:scalar-lap-fredholm} implies that we can find $u_0\in W^{k,p}_{(2-n)^+}$
with $-\Delta_g u_0 = \mu$. We observe that
\[
\overline{\Delta} u_0 = \underbrace{\Delta_g u_0 +(\overline\Delta -\Delta_g) u_0}_{=: F},
\]
where the first term contained in $F$ is smooth and compactly supported, and the 
second term (and hence $F$ itself) belongs to $W^{k-2,p}_{(\tau-n)^+}$ as a consequence of equation \eqref{eqn:lap-diff} and Lemma \ref{lem:mult-basic}. 
Lemma 1 of \cite{bieri_brill_2025} implies that, for each $R > 0$, there exists a function $w\in W^{k,p}_{(2-n+\tau)^+}$ 
with $\overline{\Delta} w= F$ on $E_R$. Thus $h:=u_0-w$ satisfies $\overline{\Delta} h = 0$
on $E_R$. Since $h$ decays as $o(1)$, the classical multipole expansion for harmonic functions 
shows that
\[
h = \alpha r^{2-n} + \delta h
\]
on $E_R$, where $\alpha$ is a constant and $\delta h$ is a harmonic function of order $1-n$
on $\Reals^n\setminus\{0\}$.

Now choose $\chi$ to be a cutoff function that equals $0$ on $B_{R/2}$ and 1 on $E_R$. We set
\[
v = \chi w + \chi \delta h,
\]
and we observe that the decay rates of the two terms imply that  $v\in W^{k,p}_{(2-n+\eta)^+}$
with $\eta = \max\{\tau,-1\}$.  Since $u_0 = \alpha r^{2-n} + v$ on $E_R$, it remains only to determine the value of $\alpha$.  We integrate by parts and invoke the fast decay of $v$ to find that
\[
1 = \int_{\Reals^n} (-\Delta_g u_0) \; dV_g = \lim_{R' \to \infty} \int_{\partial B_{R'}} -(\nabla_a u_0) \nu^a \; d A = \lim_{R' \to \infty} \int_{\partial B_{R'}} -\frac{\partial u_0}{\partial r} \; d \overline{A} = \alpha (n-2)|S^{n-1}|,
\]
as desired.
\\
\end{proof}
As in Proposition \ref{prop:fastdecay} with the momentum carrier vector fields, this result now allows us to peel off a mass term and control the remainder's decay:

\begin{corollary}\label{cor:Lap-next-range}
    Suppose that $g_{ab}$ satisfies Assumption \ref{assumption:main} with parameters $k$, $p$, and $\tau$.
    Suppose that $1-n<\delta<2-n$, that we choose $f\in W^{k-2,p}_{\delta-2}$, and that we set $c = -\int_{\Reals^n} f dV_g$.  
    If $u \in W^{k,p}_{(2-n)^+}$ satisfies $\Delta_g u = f$, then
    \[
    u = c u_0 + v,
    \]
    where $u_0$ is the function from Proposition \ref{prop:u0} and 
    $v\in W^{k,p}_{\delta}$.
\end{corollary}
\qed

While this reads as a statement about the Poisson equation, it can be applied to the Lichnerowicz equation \eqref{eqn:lichnerowicz-u} as follows: 
If $1-n < \tau < 2-n$ and $u \in W^{k,p}_{(2-n)^+}$ solves \eqref{eqn:lichnerowicz-u} with $K_{ab} \in W^{k-1,p}_{(1-n)^+}$, then $\Delta_g u = f$ for $f \in W^{k-2,p}_{\tau-2}$ defined by
\be \label{eqn:f-lich}
f = \frac{1+u}{q_n+2}R(g) - \frac{(1+u)^{-q_n-1}}{q_n+2} |K|_g^2.
\ee
Combined with Proposition \ref{prop:u0}\eqref{part:u0exp}, Corollary \ref{cor:Lap-next-range} now ensures that there is a constant $C$ and a function $v \in W^{k,p}_\tau$ such that
\[ \varphi := 1 + u = 1 + \frac{C}{r^{n-2}} + v.\]
This yields the expected mass term $Cr^{2-n}$ plus a remainder with the prescribed decay. We observe that unlike the ADM momentum components in Proposition \ref{prop:fastdecay} and Theorem \ref{thm:momentum-summary}, the analogous integral $\int_{\Reals^n} f dV_g$ is not computable directly from the seed data here, so this result may not be leveraged to prescribe the ADM mass of the resulting initial data set a priori.

It remains to discuss the more delicate case of allowing threshold decay $O(r^{2-n})$ terms in $g_{ab}$, which we handle in a manner directly analogous to Section \ref{sec:momentum}. Moving forward, we slightly alter Assumption \ref{assumption:main} to assume that $g$ is asymptotically Euclidean of class $W^{k,p}_{(2-n)^+}$, still imposing that $k > 1+n/p$. For any $K_{ab} \in W^{k-1,p}_{(1-n)^+}$, Corollary \ref{cor:lichnerowicz-existence-fast} still guarantees the existence and uniqueness of a solution $u \in W^{k,p}_{(2-n)^+}$ to equation \eqref{eqn:lichnerowicz-u} as long as $g$ is Yamabe positive, but we wish to ensure that no logarithmic terms arise and that we can control the decay of any remainder beyond the $O(r^{2-n})$ terms. We supplement Corollary \ref{cor:Lap-next-range} with the following (compare Proposition \ref{prop:W-threshold-decomp}):

\begin{lemma}
\label{lem:borderline} 
Suppose that $g_{ab}$ is asymptotically Euclidean of class  $W^{k,p}_{(2-n)^+}$ with $k > 1+n/p$, that $\zeta \in W^{k-2,p}(S^{n-1})$ satisfies $\int_{S^{n-1}} \zeta\; dV_{S^{n-1}} = 0$, and that $u\in W^{k,p}_\delta$ for some $\delta<0$ satisfies
\[
\Delta_g u = \zeta r^{-n}
\]
on the exterior region $E_R$ for some $R > 0$. Then there is a function $v\in W^{k,p}_{(1-n)^+}$
and a function $\omega\in W^{k,p}(S^{n-1})$ such that, on $E_R$,
\[
u = \omega r^{2-n} + v.
\]  
\end{lemma}
\begin{proof}
    Since $\int_{S^{n-1}} \zeta \; dV_{S^{n-1}}=0$ we can find $\omega \in W^{k,p}(S^{n-1})$ with $\Delta_{S^{n-1}} \omega = \zeta$.  This function $\omega$ is unique up 
    to an additive constant; we fix one such function here.
    
    A straightforward computation in polar coordinates shows 
    that $\overline\Delta (\omega r^{2-n}) = \zeta r^{-n}$,
    and we define $w = \chi \omega r^{2-n}$, where $\chi$ is a cutoff function equal
    to $1$ outside $B_R$ and vanishing on $B_{R/2}$.
    Let $p = u-w$, so that on $E_R$, $p$ satisfies
    \[
    \Delta_g p = -(\Delta_g - \overline\Delta)w \in W^{k-2,p}_{(2-2n)^+} \subset W^{k-2,p}_{(-1-n)^+}.
    \]
    Corollary \ref{cor:Lap-next-range} implies that $p = a u_0 + q$, where $a$ is a constant, $u_0$ is the 
    function from Proposition \ref{prop:u0}, and $q\in W^{k,p}_{(1-n)^+}$.
    By the assumed decay rate for $g_{ab}$, it follows that $u_0$ has the form
    $\alpha r^{2-n}  + W^{k,p}_{(1-n)^+}$ on $E_R$. We conclude that $u = p+w$ has the 
    claimed form on $E_R$.
    \\
\end{proof}
In the absence of the integral condition $\int_{S^{n-1}} \zeta \; dV_{S^{n-1}} = 0$, logarithmic terms generally arise in this lemma's conclusion. To combine this result with Corollary \ref{cor:Lap-next-range}, we must check that the leading order $O(r^{-n})$ term in equation \eqref{eqn:f-lich} coming from the scalar curvature $R(g)$ satisfies this integral condition if $g_{ab}$ has the following natural form (compare to Lemma \ref{lem:div-obstr-vanish}):

\begin{lemma} \label{lem:scalar-expansion}
Suppose that $g_{ab}$ is asymptotically Euclidean of class $W^{k,p}_{(2-n)^+}$ with $k > 1+n/p$ and admits the decomposition
\be \label{eqn:type-a-unphysical}
g_{ab} = \delta_{ab} + r^{2-n}\beta_{ab} + \gamma_{ab},
\ee
on $E_R$ for some $R>0$, where $\beta_{ab} \in W^{k,p}_{\mathrm{loc}}(\Reals^n \setminus \{0\})$ is homogeneous of degree zero and $\gamma_{ab}\in W^{k,p}_{\eta}$ for some $\eta<2-n$.  
Then the scalar curvature admits the decomposition
\[
R(g) = R_{\mathrm{slow}} + R_{\mathrm{fast}},
\]
on $E_R$, where
$
R_{\mathrm{slow}} = \zeta r^{-n}
$
with $\zeta\in W^{k-2,p}(S^{n-1})$, and
$
R_{\mathrm{fast}} \in W^{k-2,p}_{\eta-2} \cup W^{k-2,p}_{(2-2n)^+}.
$
Moreover, $\int_{S^{n-1}} \zeta \; dV_{S^{n-1}} = 0$, and $\zeta\equiv 0$ if $\beta_{ab}$
is a multiple of $\delta_{ab}$.
\end{lemma}
\begin{proof}
If we define $F_{ab} := \beta_{ab} r^{2-n}$, then a computation shows that
\[
R_g = \underbrace{\frac12 \left[2\partial^a \partial^b F_{ab} - 2 \overline{\Delta} F^a_a\right]}_{R_{\mathrm{slow}}} + R_{\mathrm{fast}}
\]
where $R_{\mathrm{fast}}$
is dominated by terms of the form $\partial F\partial F$, $F\partial^2F$,
and $\partial^2\gamma$, and hence
lies in $W^{k-2,p}_{\rho}$ for any $\rho$ satisfying
$\rho> 2-2n$ and $\rho \ge \eta-2$. 
We note that here, and for the remainder of the proof, we perform the index operations using the metric $\delta_{ab}$.

We decompose $\beta_{ab} = m_{ab} + v_a n_b + n_a v_b + q n_a n_b$,
where $n^a = x^a/r$ is the radial unit vector, where 
$m_{ab}$ and $v_b$ each annihilate $n^a$, and where each of 
$m_{ab}$, $v_a$, and $q$ are functions of the angle only.
Using the relations
\begin{align}
\partial_a r &= n_a, \notag \\
\partial_a n_b& =  \frac{1}{r}\left[\delta_{ab}-n_an_b\right], \notag \\
\partial^a\left(\frac{q n_a}{r^{n-1}}\right) &= 0, \notag \\
\partial^a v_a &= \frac{1}{r}\nabla^a_{S^{n-1}} v_a, \notag \\
\partial^a(m_{ab})&= \frac{1}{r}\left[ \nabla^a_{S^{n-1}} m_{ab} - \delta^{ac}m_{ac} n_b\right], \notag
\end{align}
we find that a computation shows that
\[
\partial^a \partial^b D_{ab} =
\frac1{r^n}\left[ \nabla_{S^{n-1}}^a \nabla_{S^{n-1}}^b m_{ab} + 2 \nabla^a_{S^{n-1}}v_a \right].
\]
Similarly, we have
\[
\overline{\Delta} D^a_a  = \frac{1}{r^n} \Delta_{S^{n-1}} (B^a_a).
\]
As a consequence, we obtain
\[
R_{\mathrm{slow}} = \frac{1}{r^n} \zeta,
\]
where $\zeta\in W^{k-2,p}(S^{n-1})$ satisfies the condition
$\int_{S^{n-1}} \zeta \; dV_{S^{n-1}} = 0$. 

In the case $\beta_{ab}=\delta_{ab}$ we have that $B^{a}_a$ is constant, $v_a$ is zero,
and $m_{ab}$ is the metric on the sphere, and hence $\zeta=0$.
\\
\end{proof}

Once again, we compile the results of the present section into the following theorem summarizing the prescription of the asymptotics of the conformal factor $\varphi = 1+u$ in the conformal method (compare Theorem \ref{thm:momentum-summary}):

\begin{theorem}
\label{thm:lichnerowicz-summary}
    Suppose that $g_{ab}$ satisfies Assumption \ref{assumption:main} with parameters $k$, $p$, and $\tau$; suppose that $K_{ab} \in W^{k-1,p}_{\tau-1} \cup W^{k-1,p}_{(1-n)^+}$; and suppose that $Y_g > 0$. Then the following hold:
    
    \begin{enumerate}[(a)]
    \item If $2-n<\tau<0$, then there exists a unique $u \in W^{k,p}_{\tau}$ for which $\varphi = 1+u$ solves the Lichnerowicz equation \eqref{eqn:lichnerowicz}.
    \label{part:phi-slow-falloff}

    \item If $1-n<\tau<2-n$, then the unique $u \in W^{k,p}_{(2-n)^+}$ for which $\varphi = 1+u$ solves the Lichnerowicz equation \eqref{eqn:lichnerowicz} takes the form
    \[
    u = \frac{C}{r^{n-2}} + v
    \]
    for some constant $C$, with the remainder $v \in W^{k,p}_{\tau}$.
    \label{part:phi-fast-falloff}

    \item In the borderline case such that 
    $g_{ab}$ is asymptotically Euclidean of class $W^{k,p}_{(2-n)^+}$ and has the form
    \be \label{eqn:type-a-unphysical}
    g_{ab} = \delta_{ab} + r^{2-n}\beta_{ab} + \gamma_{ab},
    \ee
    on $E_R$ for some $R>0$, where $\beta_{ab} \in W^{k,p}_{\mathrm{loc}}(\Reals^n \setminus \{0\})$ is homogeneous of degree zero and $\gamma_{ab}\in W^{k,p}_{\eta}$ for some $1-n < \eta <2-n$, the unique $u \in W^{k,p}_{(2-n)^+}$ for which $\varphi = 1+u$ solves the Lichnerowicz equation \eqref{eqn:lichnerowicz} takes the form
    \[
    u = \omega r^{2-n} + v
    \]
    on $E_R$ for some $\omega \in W^{k,p}(S^{n-1})$, with the remainder $v \in W^{k,p}_{\eta}$.
    \label{part:phi-threshold-falloff}
    \end{enumerate}
\end{theorem}
\begin{proof}
Part \eqref{part:phi-slow-falloff} follows immediately from Theorem \ref{thm:lichnerowicz-existence-slow}. 

Part \eqref{part:phi-fast-falloff} follows immediately from the discussion following Corollary \ref{cor:Lap-next-range}.

We turn to part \eqref{part:phi-threshold-falloff}. The solution $u \in W^{k,p}_{(2-n)^+}$ to the Lichnerowicz equation \eqref{eqn:lichnerowicz-u}, 
guaranteed by Corollary \ref{cor:lichnerowicz-existence-fast}, satisfies 
$\Delta_g u = f$ with $f \in W^{k-2,p}_{(-n)^+}$ specified by equation \eqref{eqn:f-lich}. 
Here, Lemma \ref{lem:scalar-expansion} implies that the decomposition
\[ 
\Delta_g u = \zeta r^{-n} + w,
\]
holds on $E_R$, with $\zeta \in W^{k-2,p}(S^{n-1})$ satisfying $\int_{S^{n-1}} \zeta \; dV_{S^{n-1}} = 0$, and with $w \in W^{k-2,p}_{\eta-2}$. %\cap W^{k-2,p}_{(2-2n)^+}$.
Corollary \ref{cor:Lap-next-range} and Proposition \ref{prop:u0}\eqref{part:u0exp} 
imply that the solution $u_1 \in W^{k,p}_{(2-n)^+}$ to $\Delta_g u_1 = w$, which exists as a consequence of 
Proposition \ref{prop:scalar-lap-fredholm}, must satisfy the condition 
$u_1 = C r^{2-n} + u_2$ on $E_R$ for some constant $C$ and some 
$u_2 \in W^{k,p}_{\eta}$. %\cap W^{k,p}_{(1-n)^+}$. 
We now have $\Delta_g (u - u_1) = \zeta r^{-n}$ on $E_R$, so Lemma \ref{lem:borderline} 
indicates that there is a function $\omega_0 \in W^{k,p}(S^{n-1})$ and a function $u_3 \in W^{k,p}_{(1-n)^+}$ 
such that $u-u_1 = \omega_0 r^{2-n} + u_3$ on $E_R$. 
Combining these results, we have on $E_R$ that
\[
u = \omega_0 r^{2-n} + u_1 + u_3 
= (\omega_0 + C)r^{2-n} + u_2 + u_3.
\]
Identifying $\omega = \omega_0 + C$ and $v = u_2 + u_3$ completes the proof.
\\
\end{proof}

We observe that in part \eqref{part:phi-threshold-falloff}, the resulting physical metric $\tilde g_{ab} = \varphi^{q_n-2} g_{ab}$ has the leading order $O(r^{2-n})$ term beyond $\delta_{ab}$ taking the form
\[ h_{ab} = r^{2-n}\left( \beta_{ab} + (q_n-2)\omega \delta_{ab} \right),\]
so that a sufficient, but almost certainly not necessary, condition guaranteeing that this is not a constant multiple of $r^{2-n} \delta_{ab}$ is that $\beta_{ab}$ does not take the form of a function on $S^2$ multiplying $\delta_{ab}$.

\section{Constructions}
\label{sec:constructions}

In Sections \ref{sec:momentum} and \ref{sec:hamiltonian}, we have established a sequence of results which guarantee that one can use the conformal method to construct general relativistic initial data sets $(\mathbb R^n, \tilde g, \widetilde K_{ab})$ with a number of prescribed features. These results have largely been framed as statements regarding the analytical features of the solutions to equations (\ref{eqn:hamiltonian}) and (\ref{eqn:lichnerowicz-u}) and are not directly connected to the properties of the ultimate initial data set. 
We now make formal statements collating our results into theorems, Theorems \ref{thm:summary_b}-\ref{thm:summary_a} below, which summarize what can be said regarding the goal of constructing data of each of the types (CK), (B), and (A) discussed in the introduction. 
Each of these theorems follows readily from the results in Sections \ref{sec:momentum} and \ref{sec:hamiltonian}, so we mostly omit further argumentation. Our first summarizing theorem pertains to the construction of type (B) data. 
Following immediately from the well-known results stated in Theorems \ref{thm:momentum-summary}\eqref{part:K-slow-falloff} and \ref{thm:lichnerowicz-summary}\eqref{part:phi-slow-falloff}, it contains no novel content:

\begin{theorem} \label{thm:summary_b}
(Type (B)) 
    We choose $2-n < \tau < 0$, and we choose $k \in \Nats$ and $p \in (1,\infty)$ 
    satisfying $k > \ell+1+n/p$ for some $\ell \in \Nats$. 
    Suppose that $g_{ab}$ is asymptotically Euclidean of class
    $W^{k,p}_\tau$ and is Yamabe positive, that $N$ is a positive function with $\delta N = N-1 \in W^{k,p}_{\tau}$, 
    and that $A_{ab} \in W^{k-1,p}_{\tau-1}$ is traceless and symmetric. The conformal constraint 
    equations \eqref{eqn:lichnerowicz}-\eqref{eqn:conformalmomentum} then admit a unique solution for
    the vector field $W^a \in W^{k,p}_\tau$ and the positive function $\varphi$ with $\varphi - 1 \in W^{k,p}_\tau$.
    The corresponding initial data set $(\Reals^n,\, \tilde g, \, \widetilde K_{ab})$ satisfies the Einstein (vacuum, constant mean curvature) constraint equations \eqref{eqn:hamiltonian}-\eqref{eqn:momentum} along with the type (B) asymptotic conditions:
    \begin{align}
    \tilde g_{ij} & = \delta_{ij} + o_{\ell+1}(r^\tau), \notag \\
    \widetilde K_{ij} & = o_{\ell}(r^{\tau-1}). \notag
    \end{align}
\end{theorem}
\qed

To illustrate this result with $\tau = 5/2-n$, one first chooses $g_{ab}$ to be Yamabe positive and asymptotically Euclidean of class $W^{k,p}_{5/2-n}$, as well as $A_{ab} \in W^{k-1,p}_{3/2-n}$, so that $Z_a = \Div_g \left( \frac{1}{2N} A_{ab} \right) \in W^{k-2,p}_{1/2-n}$. 
Proposition \ref{prop:vec-lap-properties}\eqref{part:iso} ensures that one may uniquely solve equation \eqref{eqn:conformalmomentum}, that is $(P_{g,N} W)_a = Z_a$, for $W^a \in W^{k,p}_{5/2-n}$. 
One now has $K_{ab} = \frac{1}{2N}(A_{ab} - (\ck W)_{ab}) \in W^{k-1,p}_{3/2-n}$, so Theorem \ref{thm:lichnerowicz-existence-slow} ensures that we may find a conformal factor $\varphi$ solving equation \eqref{eqn:lichnerowicz} and satisfying $\varphi-1 \in W^{k,p}_{5/2-n}$. 
For $n = 3$ and for sufficiently large $k$ ($k > 3+n/p$), the conformally transformed metric $\tilde g = \varphi^4 g$ has the form of equation \eqref{eqn:typebg}, and the physical second fundamental form $\widetilde K_{ij} = \varphi^{-2} K_{ij} \in W^{k,p}_{5/2-n}$ has the form \eqref{eqn:typebk}, which corresponds to type (B) data.

In attempting to repeat this procedure for either type (CK) or type (A) data, defined in equations \eqref{eqn:typeckg}-\eqref{eqn:typeckk} and \eqref{eqn:typeag}-\eqref{eqn:typeak}, with all the falloff weights shifted down by one, one finds that $Z_a \in W^{k-2,p}_{-1/2-n}$, so that the desired domain of $P_{g,N}$ now has $\delta = 3/2 - n < 2 - n$, and we cannot simply invoke Proposition \ref{prop:vec-lap-properties}\eqref{part:iso}. 
This obstruction is not surprising: Proposition \ref{prop:vec-lap-properties} indicates that we should expect the solution $W^a$ to equation (\ref{eqn:pmomentum}) to be contained in $W^{k,p}_{(2-n)^+}$, ultimately corresponding to $\widetilde K_{ab} \in W^{k-1,p}_{(1-n)^+}$ and allowing an $O(r^{1-n})$ term in $\widetilde K_{ab}$ which gives rise to the initial data set's linear ADM momentum. 
The same obstruction arises in determining $\varphi$: The conclusion of Theorem \ref{thm:lichnerowicz-existence-slow} is limited due to the generic appearance of an $O(r^{2-n})$ term in $\varphi$, corresponding to the mass term in the metric and screening out control of any remaining decay. While the problem of extracting the $O(r^{2-n})$ term from $\varphi$ to control the remaining decay has been treated in the literature (see Proposition 1 of \cite{bieri_brill_2025} or Theorem 1.17 of \cite{bartnik1}), previously existing results are insufficient to guarantee the emergence of type (A) data in which the mass term $h_{ij}$ in the metric is anisotropic, or to fully understand the $O(r^{1-n})$ term in the second fundamental form. As demonstrated in Theorems \ref{thm:summary_ck} and \ref{thm:summary_a} below, we are nevertheless able to treat both of these concerns.

Our second summarizing theorem, relevant to the construction of type (CK) data, partially overcomes these hurdles, with the remainder being left for our final theorem. 
Theorem \ref{thm:summary_ck} contains the novel result that the ADM momentum components of the constructed initial data set may always be fully anticipated, and hence prescribed as a consequence of one's choice of the seed data, thereby allowing one to build initial data sets in the center of mass frame with nontrivial $\widetilde K_{ab}$ if desired. 
The following is a consequence of Theorem \ref{thm:momentum-summary}\eqref{part:K-fast-falloff} (and the immediately following discussion) and Theorem \ref{thm:lichnerowicz-summary}\eqref{part:phi-fast-falloff}:

\begin{theorem} \label{thm:summary_ck}
(Type (CK)) 
    We choose $1-n < \tau < 2-n$, and we choose $k \in \Nats$ and $p \in (1,\infty)$ 
    which satisfy $k > \ell+1+n/p$ for some $\ell \in \Nats$. 
    Suppose that $g_{ab}$ is asymptotically Euclidean of class
    $W^{k,p}_\tau$ and is Yamabe positive, that $N$ is a positive function with $\delta N = N-1 \in W^{k,p}_{\tau}$, 
    and that $A_{ab} \in W^{k-1,p}_{\tau-1}$ is traceless and symmetric. The conformal constraint 
    equations \eqref{eqn:lichnerowicz}-\eqref{eqn:conformalmomentum} then admit a unique solution for
    the vector field $W^a \in W^{k,p}_{(2-n)^+}$ and the positive function $\varphi$ with 
    $\varphi - 1 \in W^{k,p}_{(2-n)^+}$. The corresponding initial data set 
    $(\Reals^n,\, \tilde g, \, \widetilde K_{ab})$ has ADM momentum components 
    \be \label{eqn:momentum-integrals}
    p_j := \mathcal P_j(\widetilde K_{ab}) 
    = \mathcal G_j \left(\frac{1}{2N}A_{ab} \right) 
    = -\mathcal R_j\left(\frac{1}{2N}A_{ab} \right)
    := -\int_{\Reals^n} \nabla^b \left( \frac{1}{2N} A_{ab}\right) k^a_{(j)}\; dV_g,
    \ee 
    and there is a constant $C$ such that
    \begin{align}
    \tilde g_{ij} & = \left(1+\frac{C}{r^{n-2}}\right)\delta_{ij} + o_{\ell+1}(r^\tau), \notag \\
    \widetilde K_{ij} & = \frac{2n(n-2) C_n}{r^{n-1}} 
    \left[ p_{i} n_j + p_{j} n_i + ((n-2) n_i n_j - \delta_{ij}) p_{k}n^k \right] + o_\ell(r^{\tau-1}). \notag
    \end{align}
\end{theorem}
\qed
 
We recall that $k_{(j)}^a$ denotes the vector fields of Proposition \ref{prop:kernel}, which comprise a particular basis of the kernel of $P_{g,N}$ on $W^{k,p}_{0^+}$. In order to use this result to prescribe the ADM momentum components in practice, one should fix $g$ and $N$ according to the above hypotheses, numerically construct the kernel vector fields $k^a_{(j)}$, and compute each of the integrals \eqref{eqn:momentum-integrals} for $n$ different choices of $A_{ab}$. This procedure generically produces $n$ independent momentum vectors, resulting in a basis that one can use to construct any desired momentum vector by specifying an appropriate linear combination of one's choice of the $A_{ab}$'s. To most simply construct an initial data set in the center of mass frame, one should choose $A_{ab} = 0$. To construct an initial data set in the center of mass frame which is not time-symmetric, one should compute the momentum vector for one more choice of $A_{ab}$ (for a total of $n+1$ choices, necessarily producing a linearly dependent set of momentum vectors), thereby allowing one to build a nontrivial linear combination with zero momentum.

Our third and final summarizing theorem is relevant to the construction of type (A) data. 
It contains the novel result that an anisotropic mass term can be built into the metric $\tilde g_{ij}$ of one's initial data set in practice, providing an intermediate regime between type (B) and type (CK) data. 
It further indicates that a Euclidean transverse traceless tensor which is homogeneous of degree $1-n$ can be built into $\widetilde K_{ij}$ as its leading order term if desired; this is the momentum analogue of the anisotropic mass term allowed for in $\tilde g_{ij}$.
The following statement is a consequence of Theorem \ref{thm:momentum-summary}\eqref{part:K-threshold-falloff} and Theorem \ref{thm:lichnerowicz-summary}\eqref{part:phi-threshold-falloff} (and the immediately following discussion):

\begin{theorem} \label{thm:summary_a}
(Type (A)) 
    We choose $1-n < \tau < 2-n$, and we choose $k \in \Nats$ and $p \in (1,\infty)$ 
    satisfying $k > \ell+1+n/p$ for some $\ell \in \Nats$. 
    Suppose that $N$ is a positive function with $\delta N = N-1 \in W^{k,p}_{(2-n)^+}$. Suppose that $g_{ab}$ is asymptotically Euclidean of class $W^{k,p}_{(2-n)^+}$, is Yamabe positive, and takes the form
    \[
    g_{ij} = \delta_{ij} + r^{2-n}\beta_{ij} + \gamma_{ij}
    \]
    on $E_R$ for some $R>0$, where $\beta_{ij} \in W^{k,p}_{\mathrm{loc}}(\Reals^n \setminus \{0\})$ is homogeneous of degree zero 
    and $\gamma_{ij}\in W^{k,p}_{\tau}$. 
    Suppose that $A_{ab}\in W^{k-1,p}_{(1-n)^+}$ has the form
    \be
    A_{ab} = r^{1-n} B_{ab} + C_{ab}
    \ee
    on $E_R$, where $B_{ab} \in W^{k-1,p}_{\mathrm{loc}}(\Reals^n\setminus\{0\})$ is homogeneous of degree zero and $C_{ab}\in W^{k-1,p}_{\tau-1}$, and both are traceless and symmetric. The conformal constraint 
    equations \eqref{eqn:lichnerowicz}-\eqref{eqn:conformalmomentum} admit a unique solution for
    the vector field $W^a \in W^{k,p}_{(2-n)^+}$ and the positive function $\varphi$ with 
    $\varphi - 1 \in W^{k,p}_{(2-n)^+}$. The corresponding initial data set 
    $(\Reals^n,\, \tilde g, \, \widetilde K_{ab})$ has the following ADM momentum components:
    \[
    \mathcal P_j(\widetilde K_{ab}) 
    = \mathcal G_j \left(\frac{1}{2N}A_{ab} \right) 
    := \int_{\Reals^n} \frac{1}{2N} A_{ab} \nabla^b k^a_{(j)}\; dV_g.
    \] 
    There is a homogeneous of degree $2-n$ symmetric tensor $h_{ij}$ and a homogeneous of degree $1-n$ Euclidean transverse traceless tensor $D_{ij}$ such that
    \begin{align}
    \tilde g_{ij} & = \delta_{ij} + h_{ij} + o_{\ell+1}(r^\tau), \notag \\
    \widetilde K_{ij} & = D_{ij} + o_\ell(r^{\tau-1}). \notag
    \end{align}
    Moreover: if $\beta_{ab} \neq \omega \delta_{ab}$ for any $\omega \in W^{k,p}(S^{n-1})$, 
    then $h_{ij}$ is not a constant multiple of $r^{2-n} \delta_{ij}$.
    \end{theorem}
\qed

We remark that the discussion at the end of Section \ref{sec:momentum} further implies that, with an appropriate choice of $A_{ab} \in W^{k-1,p}_{(1-n)^+}$, one may arrange that the leading order term in $\widetilde K_{ij}$ can be {\it any} given Euclidean transverse traceless tensor $D_{ij}$ that is homogeneous of degree $1-n$. That is, this term can be prescribed, and it could readily be made nontrivial even with zero ADM momentum.

Taken together, Theorems \ref{thm:summary_b}-\ref{thm:summary_a} answer the question of how one can practically construct initial data sets of each of the types (CK), (B), and (A) (and a bit beyond) using the formalism of the conformal method. These theorems indicate that one can control the initial data set's ADM momentum and the leading (or next-to-leading) order decay rates of the physical metric $\tilde g_{ij}$ and the second fundamental form $\widetilde K_{ij}$, and that one may build an anisotropic mass term into $\tilde g_{ij}$, or the analogous term into $\widetilde K_{ij}$, if desired. 
As an application, these results guarantee that a simple example numerically constructed below is of type (A), allowing us to place the strongest known constraints (cf.\@ \cite{bieri_brill_2025}) on the collection of spacetimes admitting the symmetry suggested by the antipodal matching conjecture of \cite{strominger}.

\section{A Numerical Example}
\label{sec:numerics}

In this section, we present a numerically computed Brill wave solution of type (A) using the same numerical method that is used in \cite{bieri_brill_2025}. Brill waves are initial data sets on $\Reals^3$ for which $\widetilde K=0$ and $\tilde g$ is axisymmetric, first investigated in \cite{Brill1}.  The momentum constraint is then trivially satisfied, while the Lichnerowicz equation \eqref{eqn:lichnerowicz} becomes
\be
-8 {\Delta_g}\varphi + R(g) \varphi = 0.
\label{Brill1}
\ee

We find it convenient to perform the numerical computations in spherical polar coordinates $(r,\theta,\phi)$.  In these coordinates, the line element of the unphysical metric $g$ takes the form 
\be
d {s^2} = {e^{2q}} (d {r^2} + {r^2} d {\theta ^2}) + {r^2} {\sin^2} \theta d {\phi^2}.
\ee
Here, in keeping with axisymmetry, the function $q$ depends only on $r$ and $\theta$.  
Equation \eqref{Brill1} then takes the form
\be
- \overline \Delta \varphi + S \varphi = 0,
\label{Brill2}
\ee
where $\overline \Delta$ is the flat Laplacian and $S$ is given by 
\be
S = - {\frac 1 4} \left ( {\frac {{\partial ^2}q} {\partial {r^2}}} + {\frac 1 r} \, {\frac {\partial q} {\partial r}} + {\frac 1 {r^2}} \, {\frac {{\partial ^2}q} {\partial {\theta ^2}}} \right ). 
\ee
 
We define $F := \ln \varphi$ and decompose $F={F_1}+{F_2}$, with $F_1$ a solution to 
\be
\overline \Delta {F_1} = S.
\label{F1eqn}
\ee
Equation (\ref{Brill2}) now becomes
\be
\overline \Delta {F_2} = - {\vec \nabla}F \cdot {\vec \nabla}F,
\label{F2eqn}
\ee
which we seek to solve for $F_2$. 
Given a numerical method to invert the flat Laplacian, we first solve equation (\ref{F1eqn}) and proceed to solve equation (\ref{F2eqn}) by iteration.  
That is, having found $F_1$, we make the initial guess of zero for $F_2$ and repeatedly solve equation (\ref{F2eqn}) for an improved version of $F_2$, where the right hand side of equation (\ref{F2eqn}) is computed using the previous version of $F_2$.  At each step, the current version of $F_2$ is stored, with the final version used to compute $F$.

Our numerical method to invert the flat Laplacian relies on the standard Green's function for axisymmetric functions (see e.g., \cite{Jackson}) where we compute all integrals numerically. Explicitly, for equation (\ref{F1eqn}) we have 
\be
{F_1} = {\sum _{\ell =0} ^\infty} \left [ {g_\ell}(r) {r^{-(\ell+1)}} + {h_\ell}(r) {r^\ell}\right ] {P_\ell}(\cos \theta ),
\label{inverselaplacian}
\ee
where ${g_\ell}(r)$ and ${h_\ell}(r)$ are specified by 
\be
{g_\ell}(r) = - {\textstyle {\frac 1 2}} {\int _0 ^r} {{\tilde r}^{\ell+2}} d {\tilde r} {\int _0 ^\pi} \sin \theta \, d \theta \, S({\tilde r},\theta){P_\ell}(\cos \theta),
\label{gl}
\ee
\be
{h_\ell}(r) = - {\textstyle {\frac 1 2}} {\int _r ^\infty} {{\tilde r}^{1-\ell}} d {\tilde r} {\int _0 ^\pi} \sin \theta \, d \theta \, S({\tilde r},\theta){P_\ell}(\cos \theta),
\label{hl}
\ee
and $P_\ell$ are the Legendre polynomials. We note that equations (\ref{inverselaplacian}-\ref{hl}) provide an analytic expression for the inverse of the Laplacian. Thus the only ``numerical'' features of this numerical method are (1) the integrals in equations (\ref{gl}-\ref{hl}) are computed numerically, and (2) the use of a computer makes it straightforward to include a large enough number of values of $\ell$ for accuracy and a large enough number of iterations of the iterative procedure for convergence to numerical accuracy. In both cases, ``large enough" is practically assessed according to whether both sides of the relevant PDE (either the appropriate Poisson equation when choosing a maximum $\ell$ value to invert the Laplacian, or equation \eqref{F2eqn} when deciding to stop iterating) agree, in an $L^2$ sense, within an acceptably small tolerance.

To assess whether our constructed initial data set conforms to the antipodal symmetry suggested by the conjecture of \cite{strominger}, we examine the behavior of the curvature component $\rho$, which is the principal curvature corresponding to the plane that contains the gradient of $r$ and the normal to the initial data surface. For Brill wave initial data, $\rho$ is given asymptotically (as $r \to \infty$) by the following expression:
\bea
{r^2} \rho &=& - {r^2} {\frac {{\partial ^2}q} {\partial {r^2}}} - 
{\frac {{\partial ^2}q} {\partial {\theta^2}}} - \cot \theta {\frac {\partial q} {\partial \theta}} + 2 {r^2} {\frac {\partial q} {\partial r}} {\frac {\partial F} {\partial r}} - 2 {\frac {\partial q} {\partial \theta}} {\frac {\partial F} {\partial \theta}}
\nonumber
\\
&& - 4 {r^2} {\frac {{\partial ^2}F} {\partial {r^2}}} - 4 r {\frac {\partial F} {\partial r}} - 2 {\frac {{\partial ^2}F} {\partial {\theta^2}}} - 2 \cot \theta {\frac {\partial F} {\partial \theta}} - 4 {{\left ( {\frac {\partial F} {\partial \theta}} \right ) }^2}.
\eea

\begin{figure}[t!]	
\centering
\includegraphics[width=0.7\linewidth]{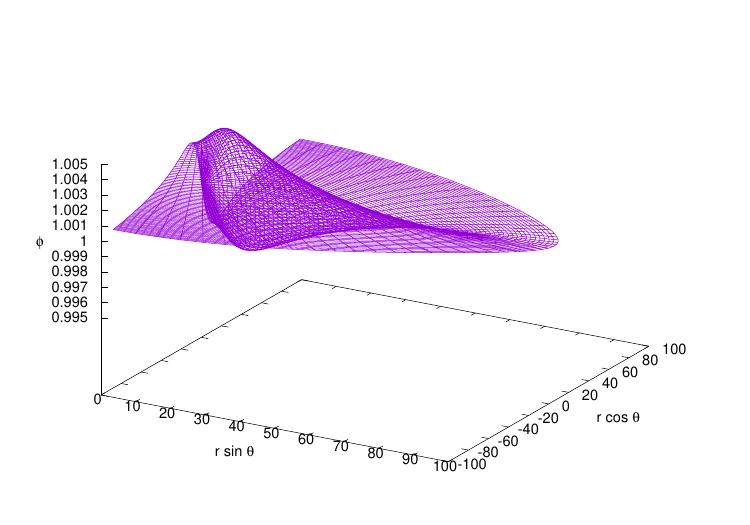}
	\caption{ A plot of $\varphi$, computed with $q$ specified by equation (\ref{qformula2}) with ${a_0}=1, \, {r_0}=10, \, \gamma = 2$.}			\label{phifig}
\end{figure}

In order to produce an example that is not antipodally symmetric, we consider $q$ of the form
\be
q={a_0}{r^3}\cos (\theta) {\sin ^2}(\theta) {{({r^2}+{r_0^2})}^{-\gamma}},
\label{qformula2}
\ee
where ${a_0}, \, {r_0}, \,$ and $\gamma$ are constants.  
In order to obtain type (A) behavior, we choose $\gamma = 2$. The fact that the resulting initial data set is of type (A) is now implied by Theorem \ref{thm:summary_a}.  We have chosen this particular $q$ because it is a minimal example which yields a type (A) metric that is not antipodally symmetric while satisfying the regularity and axisymmetry conditions $q = \frac{\partial q}{\partial \theta} = 0$ at $\theta = 0$ and $\theta = \pi$ (see \cite{Brill1}).
Figure \ref{phifig} graphs the numerically computed $\varphi$ for the case
${a_0}= 1, {r_0}=10, \gamma = 2$ up to the radius of $r=100$.  Figure \ref{cfig} plots ${r^3}\rho$ for two different values of $r$: $r=10000$ and $r=15000$.   The fact that the two curves agree corroborates the expectation that the asymptotic behavior of $\rho$ is $\rho \propto {r^{-3}}$.

We observe in Figure \ref{cfig} that $\rho$ is not antipodally symmetric.  This is relevant to the conjecture of Strominger \cite{strominger} that in the limit of early time at null infinity there is a symmetry involving the combination of time reflection and antipodal mapping.  Since Brill wave initial data sets have zero extrinsic curvature, it follows that their time evolution has time reflection symmetry.  The evolution of the Brill wave initial data set whose features are plotted in Figures \ref{phifig} and \ref{cfig} therefore cannot satisfy the symmetry suggested by the conjecture of \cite{strominger}.  This result should not be regarded as a counterexample to this conjecture, but rather as a statement regarding the sort of asymptotic flatness under which this conjecture can hold: Specifically, our results suggest that the conjecture of \cite{strominger} may not generally hold for metrics of type (A). We remark that the physical conjecture regards certain charges at null infinity which require $\rho$ to decay like $O(r^{-3})$ towards {\it null} infinity in order to exist. Although the initial data set constructed here has $\rho = O(r^{-3})$ towards {\it spatial} infinity (this is true of any type (A) data), we have not proven that it does not pick up leading-order log terms under time evolution by the vacuum Einstein equation. We suspect that no such terms develop, but a proof requires an in-depth analysis that we hope to provide in future work.

\begin{figure}[t!]	
\centering
\includegraphics[width=0.7\linewidth]{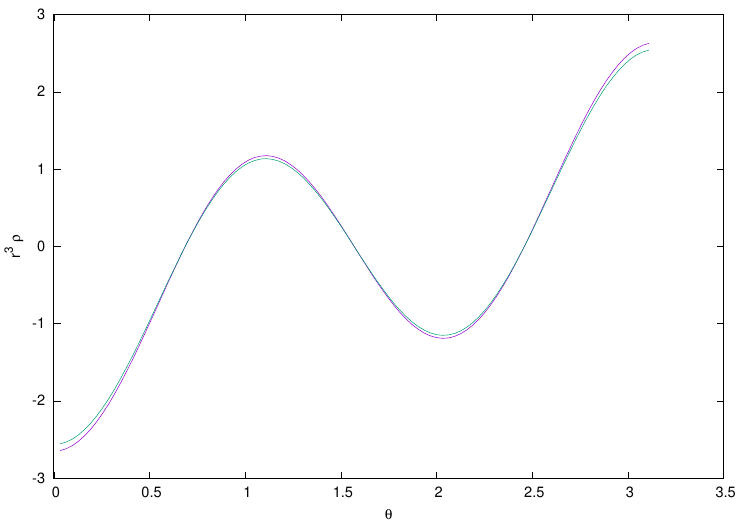}
	\caption{ Superimposed plots of ${r^3}\rho$ as a function of $\theta$, once at $r=10000$ and once at $r=15000$, computed with $q$ specified by equation (\ref{qformula2}) with ${a_0}=1, \, {r_0}=10, \, \gamma = 2$.	}	\label{cfig}
\end{figure}

\section{Conclusion}
\label{sec:conclusion}

In this work, we have thoroughly investigated the issue of constructing maximal solutions $(\mathbb R^n, \tilde g, \widetilde K_{ab})$ to the vacuum constraint equations \eqref{eqn:hamiltonian}-\eqref{eqn:momentum} via the conformal method with prescribed asymptotic behavior. 
While it has long been known that one can construct unique solutions for any choice of seed data $g_{ab}$, $A_{ab}$, and $N$ taken in appropriate function spaces, we have clarified the extent of the technical control which can be exerted over the asymptotics of these solutions. 
At a broad level, we have proceeded by first characterizing the solutions to the conformal constraint equations \eqref{eqn:lichnerowicz}-\eqref{eqn:conformalmomentum} explicitly at leading order, using the ``momentum carrier" vector fields $W^a_{(j)}$ and the ``mass carrier" function $u_0$ in the event of rapid decay in the seed data, or using understood solutions from analogous problems involving the flat operators $\overline P$ and $\overline \Delta$ in the event of threshold decay corresponding to the natural forms \eqref{eqn:A-threshold-decomp} and \eqref{eqn:type-a-unphysical} for $A_{ab}$ and $g_{ab}$. 
Generally, we have shown that peeling off this leading-order behavior of the solutions renders the equations solvable in a space with faster decay, asserting control of the decay rate of subleading corrections.
Theorems \ref{thm:momentum-summary} and \ref{thm:lichnerowicz-summary} summarize these results at the more abstract level of the vector field $W^a$ and the conformal factor $\varphi$, which are directly solved for in the conformal method, while Theorems \ref{thm:summary_b}, \ref{thm:summary_ck}, and \ref{thm:summary_a} summarize these results at the level of the ultimately constructed initial data set $(\mathbb R^n, \tilde g, \widetilde K_{ab})$. We remark that the leading-order terms may be understood well enough to prescribe the initial data set's ADM momentum components, and (in working with prescribed threshold decay) to impose that the metric include an anisotropic mass term.

The results established herein ensure that, with an appropriate choice of seed data, one can construct initial data sets of any of the three types discussed in the introduction. We have illustrated this in Section \ref{sec:numerics} by numerically constructing a simple example of a type (A) spacetime guaranteed to include an anisotropic mass term, which previously available results were unable to handle. This example is constructed in such a way that its spacetime evolution cannot admit the symmetry suggested by the antipodal matching conjecture of \cite{strominger}, restricting the collection of spacetimes within which one can expect this property of physical interest to be realized. It is our hope that our results furnish the ability to further construct a wealth of initial data sets with various decay features, allowing the relativity community to probe, both numerically and analytically, how physically interesting properties of spacetimes may depend on one's choice of physically motivated asymptotic behavior.

\section*{Acknowledgements}

We thank Robert Wald and Alex Nachmann for discussions surrounding the implications of this work for the antipodal matching conjecture of \cite{strominger}. Lydia Bieri was supported by NSF grant DMS-2204182 to The University of Michigan. James Wheeler was also partially supported by NSF grant DMS-2204182 to The University of Michigan. David Garfinkle was supported by NSF grant PHY-2102914 to Oakland University. 

\appendix

\ifarxiv
\section{Fredholm Properties of the Vector Laplacian}
\else
\section{\small Fredholm Properties of the Vector Laplacian}
\fi
\label{app:vec-fredholm}

This appendix contains auxiliary technical results regarding the 
Fredholm properties of vector Laplacians.  Although the body of 
our work above is restricted to the setting of metrics with H\"older continuous first derivatives, for the sake of generality and potential other applications we now relax the condition $k>n/p+1$ of Assumption \ref{assumption:main}, assuming instead only that $k>n/p$.
\begin{assumption}\label{assumption:appendix}
    The metric $g_{ab}$ on $\Reals^n$, with $n\ge 3$, and the lapse $N$ satisfy the following two conditions for some $\tau<0$,  $k\in\Nats$, and $1<p<\infty$ with $k > n/p$:
    \begin{itemize}
        \item $g_{ab}$ is asymptotically Euclidean of class $W^{k,p}_{\tau}$.
        \item $N$ is a positive function with $\delta N = N-1 \in W^{k,p}_{\tau}$.
    \end{itemize}
\end{assumption}
The vector Laplacian is then $P_{g,N} =  \Div_g (\frac{1}{2N} \ck)$, which we abbreviate as $P$ in this section.  The model Euclidean vector Laplacian with $N=1$ is then $\overline P = \frac{1}{2} \overline \Div (\overline \ck)$. 

The natural Fredholm theory for a self-adjoint operator such as the vector Laplacian makes use of dual function spaces, and in particular distributions having a negative order of differentiability. Hence we generalize Definition \ref{def:weighted-basic} to allow weighted Sobolev spaces $W^{k,p}_\delta$ with $k \le 0$. Recall that if $k<0$, then $W^{k,p}(\Omega)$ for an arbitrary domain $\Omega \subset  \Reals^n$ is the topological dual space $W^{k,p}(\Omega) := (W^{-k,p'}_0(\Omega))^*$, with the induced dual norm, where $p'$ is the Lebesgue exponent conjugate to $p$,
$$\frac{1}{p} + \frac{1}{p'} = 1,$$ and where $W^{-k,p'}_0(\Omega)$ is the closure of $C_c^\infty(\Omega)$ in $W^{-k,p'}(\Omega)$.
\begin{definition} \label{def:weighted-full}
    Let $k\in \Ints$, $1<p<\infty$, and $\delta\in \Reals$.
    We denote by $A_r$ the annulus $\{x\in\Reals^n: r/2 < |x| < 2r\}$,
    and we let $S_r:\Reals^n \rightarrow \Reals^n$ be the scaling map $S_r(x) = rx$, so 
    $S_r(A_1) = A_r$. 
    The {\bf weighted Sobolev space} $W^{k,p}_\delta$ consists of the tempered distributions 
    $u$ on $\mathbb R^n$ such that 
    \[
 ||u||_{W^{k,p}_\delta}^p := ||u||_{W^{k,p}(B_1)}^p + \sum_{j=0}^\infty 
    2^{-j p \delta} ||S_{2^j}^* u||_{W^{k,p}(A_1)}^p < \infty.
    \]
\end{definition}
It follows from Bartnik's scaling technique \cite{bartnik1} that 
for $k\ge 0$ the norm above is equivalent to that of Definition \ref{def:weighted-basic}.
Although we do not require it, we remark that Definition \ref{def:weighted-full}
generalizes to non-integral scales of differentiabilty, and 
indeed elementary arguments show that 
$W^{k,p}_\delta$ defined here agrees with the space $f^s_{p,2,\mu}$ of \cite{triebel_spaces_1976a}
with $s=k$ and $\mu = - \delta p + sp - n$.  See also the analogous 
definition for the special case $p=2$ in \cite{maxwell_rough_2006}. General properties of these spaces in the case $k\ge 0$ can be found in \cite{bartnik1}, and their generalizations to the broader family appear in \cite{triebel_spaces_1976a} and \cite{maxwell_rough_2006}. The sole 
exception is the compactness of the embedding $W^{k,p}_\delta \hookrightarrow W^{j,p}_{\delta'}$ if $k\ge j$ and $\delta'<\delta$.  This appears in the special case $p=2$ in \cite{maxwell_rough_2006} and can be extended to general $p$ using the same technique.

The following multiplication lemma provides criteria under
which the product of elements of two weighted Sobolev spaces determines an element of a third weighted space. The proof follows from Bartnik's scaling technique and the corresponding Sobolev multiplication properties on unweighted spaces (see, for example, Theorem 2.5 of \cite{holst_scaling_2023}). See also \cite{allen_sobolevclass_2022a} Proposition 3.11, which is the equivalent result for asymptotically hyperbolic manifolds.
\begin{lemma}\label{lem:mult}
	Suppose $1<p_1,p_2,q<\infty$, $k_1,k_2,j\in\Ints$ and $\delta_1, \delta_2 \in\Reals$.  Let $r_1,r_2$ and $r$ be defined by
\[
\frac{1}{r_i} = \frac{1}{p_i} - \frac{k_i}{n}\quad i=1,2\quad\text{and}\quad
\frac{1}{r} = \frac{1}{q} - \frac{j}{n}.
\]
Pointwise multiplication of $C^\infty_c$ 
functions extends to a continuous bilinear map 
$W^{k_1,p_1}_{\delta_1}\times W^{k_2,p_2}_{\delta_2}
\rightarrow W^{j,q}_{\delta_1 + \delta_2}(\Reals)$ so long as
\begin{align*}
    \frac{1}{p_1}+\frac{1}{p_2} &\ge \frac{1}{q}, \\
    k_1+k_2&\ge 0, \\
 \min \{k_1,k_2\} &\ge j, \\
\max\left \{ \frac{1}{r_1},\frac{1}{r_2}\right\} & \le  \frac{1}{r}, \\
\text{and} \qquad \quad \frac{1}{r_1} + \frac{1}{r_2} & \le \min\left \{ 1,\frac1r\right \},
\end{align*}
with the final inequality being strict if $\min \{ 1/r_1,1/r_2,1-1/r \}=0$.
\end{lemma}
\qed

A straightforward application of Lemma \ref{lem:mult} (see Proposition 2.6 of \cite{holst_scaling_2023} for the corresponding result in unweighed spaces) establishes the following elementary mapping property of the vector Laplacian:
\[
P: W^{j,q}_\delta \to W^{j-2,q}_{\delta-2}
\]
is continuous so long as
\begin{align}
2-k &\le j \le k, \label{eqn:lap-condition1} \\
\frac1p - \frac{k}{n} &\le \frac1q -\frac{j}{n} \le \frac1{p'} - \frac{2-k}{n}. \label{eqn:lap-condition2}
\end{align}
\begin{definition}\label{def:calS}
Suppose $1<p<\infty$ and $k>n/p$. The set $\mathcal S^{k,p}$ of \textbf{compatible Sobolev indices} 
is the collection of pairs $(j,q)$ in $\Ints\times (1,\infty)$ satisfying \ref{eqn:lap-condition1}--\ref{eqn:lap-condition2}.
\end{definition}
Using the assumption $k>n/p$, it is easy to show that 
$\mathcal S^{k,p}$ is nonempty and contains $j=1$, $p=2$. This observation and related facts concerning
$\mathcal S^{k,p}$ are discussed in \cite{holst_scaling_2023}.

The dual space $(W^{j,q}_{\delta})^*$ for scalars can be identified with $W^{-j, q'}_{-n-\delta}$; see \cite{triebel_spaces_1976a} for one approach. This identification does not rely on a metric, but it is more useful for our purposes to use a metric-specific representation of the dual space that aligns with the adjoint structure of $P$.  The following result is a variation of Proposition 3.18 of \cite{allen_sobolevclass_2022a}, which proves the analogous fact in the asymptotically hyperbolic setting.  The proof follows the same technique, and is omitted.

\begin{lemma}\label{lem:dual}
Suppose $g_{ab}$ satisfies Assumption \ref{assumption:appendix} with the parameters $k$, $p$, and $\tau$.
Suppose that $\delta\in \Reals$, and that $(j,q) \in \Ints \times (1,\infty)$ satisfies $|j| \leq k$ and
\be \label{eqn:pairing-condition}
\frac1p - \frac{k}{n} \le \frac1q -\frac{j}{n} \le \frac1{p'} + \frac{k}{n}.
\ee
For any $X \in W^{-j,q'}_{-n-\delta}$, we define 
$f_X : C^\infty_c(\Reals^n) \rightarrow \Reals$
by
\begin{equation}
f_X(Y) := \int_{\Reals^n} \left<X,Y\right>_g\; dV_g.
\end{equation}
Then $f_X$ extends to a continuous map
\[
f_X : W^{j,q}_{\delta} \to \Reals.
\]
Moreover the map $X \mapsto f_X$ is a
linear isomorphism 
$W^{-j,q'}_{-n-\delta} \rightarrow (W^{j,q}_{\delta})^*$.
\end{lemma}
\qed

Using the specific identification of 
$(W^{j,q}_\delta)^*$ with $W^{-j,q'}_{-n-\delta}$ from Lemma \ref{lem:dual}, the following lemma shows that the vector Laplacian is self-adjoint:
\begin{lemma}\label{lem:self-adjoint}
Suppose $g_{ab}$ satisfies Assumption \ref{assumption:appendix} with the parameters $k$, $p$, and $\tau$. Consider Sobolev parameters $(j,q)$ and a weight $\delta$ such that $(j,q)\in \mathcal S^{k,p}$ from Definition \ref{def:calS}.  
Then $(-j,q')\in S^{k,p}$ as well, and hence
$P$ acts continuously on $W^{-j,p'}_{2-n-\delta}$.  Moreover,
for all $X\in W^{j,q}_{\delta}$ and $Y\in W^{-j,q'}_{2-n-\delta}$,
\be\label{eq:int-by-parts}
\int_{\Reals^n} \left<P X,Y\right>_g\; dV_g = \int_{\Reals^n} \left< X,P Y\right>_g\; dV_g.
\ee
\end{lemma}
\begin{proof}
The fact that $(-j,q')\in \mathcal S^{k,p}$ is an easy computation from the definition.  Equality \eqref{eq:int-by-parts} holds if $X$ and $Y$ are smooth and compactly supported.  The continuity of $P$ 
on $W^{j,q}_{\delta}$ and $W^{-j,q'}_{2-n-\delta}$ together with 
Lemma \ref{lem:dual} and the density of smooth compactly supported 
vector fields in these spaces establishes equality \eqref{eq:int-by-parts} generally.
\\
\end{proof}

We have the following elementary elliptic regularity estimate, which is a variation of Proposition 1.6 in \cite{bartnik1}.  

\begin{lemma}\label{lem:basic-reg}
    Suppose that $g_{ab}$ and $N$ satisfy Assumption \ref{assumption:appendix} with parameters $k$, $p$, and $\tau$.
    Assume $X\in W^{2-k,q'}_{\mathrm{loc}}$ so that $P X$ is
    well-defined as a distribution. 
    Suppose that for some $(j,q)\in \mathcal S^{k,p}$ from Definition \ref{def:calS} and some $\delta\in\Reals$ that
    $X\in W^{j-2,q}_{\delta}$ and that $P X = Z$ for some $Z\in W^{j-2,q}_{\delta-2}$.  Then $X\in W^{j,q}_{\delta}$ and
    \be\label{eq:basic-reg}
    ||X||_{W^{j,q}_{\delta}} \lesssim ||P X||_{W^{j-2,q}_{\delta-2}} + ||X||_{W^{j-2,q}_{\delta}}.
    \ee
\end{lemma}
\begin{proof}
Theorem 2.21 of \cite{holst_scaling_2023} implies that $X\in W^{j,q}_{\mathrm{loc}}$. To see that $X$ additionally lies in the weighted space we first
decompose
\[
P = \overline P + \sum_{|\alpha|\le 2} A^\alpha  \partial_\alpha
\]
with the matrix coefficients $A^\alpha\in W^{k-2+|\alpha|,p}_{\tau-2+|\alpha|}$. Recalling the scaling operator $S_r$ of Definition \ref{def:weighted-full} we find that for any $r=2^m$ for some $m\in\Nats$
\[
S_r^*(PX) = r^{-2} \overline P S_r^*(X) + \sum_{|\alpha|\le 2 } r^{-|\alpha|} S_r^*( A^\alpha) \partial_\alpha S_r^* X.
\]
Hence
\[
\overline P S_r^*(X) = r^2 S_r^*(PX) - r^{\tau}\sum_{|\alpha|\le 2 } r^{2-|\alpha|-\tau} S_r^*( A^\alpha) \partial_\alpha S_r^* X.
\]
Theorem 2.5 of \cite{holst_scaling_2023} on Sobolev multiplication and the embedding $\ell^q\hookrightarrow \ell^\infty$ then imply that
\[
|| r^{2-|\alpha|-\tau} S_r^*( A^\alpha) \partial_\alpha S_r^* X ||_{W^{j-2,q}(A_1)} \lesssim ||S_r^* X||_{W^{j,q}(A_1)}.
\]
Hence elliptic estimates for $\overline P$ yield
\[
||S_r^* X||_{W^{j,q}(A_1)}  \lesssim r^2 ||S^*_r PX ||_{W^{j-2,q}(A_1)} + r^{\tau} ||S_r^* X||_{W^{j,q}(A_1)} + ||X||_{W^{j-2,q}(\hat A_1)}
\]
where $\hat A^1$ is the thickened annulus $1/3 < |x| < 3$.
Since the implicit constant is independent of $r=2^m$ we can  
find $M$ sufficiently large so that if $m\ge M$ then
the term multiplied by $r^\tau$ can be absorbed into the left hand side to find that
\be\label{eq:A1-est}
||S_r^* X||_{W^{j,q}(A_1)}  \lesssim r^2 ||S^*_r PX ||_{W^{j-2,q}(A_1)} + ||S^*_r X||_{W^{j-2,q}(\hat A_1)}.
\ee
Now
\[
||S^*_{r} X||_{W^{j-2,q}(\hat A_1)} \lesssim \sum_{i=-1}^1 
2^{-i} ||S_{2^i r}^* X||_{W^{j-2,q}(A_1)}.
\]
Hence multiplying inequality \eqref{eq:A1-est} by $2^{-\delta m}$,
raising to the $q^{\rm th}$ power, and summing we find that
\[
\sum_{m\ge M} 2^{-\delta q m} ||S^*_{2^m} X||^2_{W^{j-2,q}( A_1)}
\lesssim 
\sum_{m\ge M} 2^{-(\delta - 2) q m} ||S^*_{2^m} PX ||^q_{W^{j-2,q}(A_1)}
+ 
\sum_{m\ge M-1} 2^{-\delta q m} ||S^*_{2^m} X ||^q_{W^{j-2,q}(A_1)}.
\]
Therefore
\be\label{eq:exterior-est}
\sum_{m\ge M} 2^{-\delta q m} ||S^*_{2^m} X||^2_{W^{j-2,q}(A_1)} 
\lesssim ||PX||^q_{W^{j-2,q}_{\delta-2}} + ||X||^q_{W^{j-2,q}_\delta}.
\ee
and we conclude that $X\in W^{j,q}_\delta$. Estimate \eqref{eq:basic-reg}
follows from estimate \eqref{eq:exterior-est} along with 
local regularity estimates for the finitely many terms omitted from the sum.
\\
\end{proof}

Fredholm properties of the vector Laplacian follow from those of the Euclidean model operator $\overline P$, which we establish now.

\begin{proposition}\label{prop:flatcase}
    Suppose $\delta$ is non-exceptional (i.e., $\delta$ is not an integer or $2-n<\delta<0$).  Then
    \[
    \overline{P}^{j,q}_\delta : W^{j,q}_\delta \rightarrow W^{j-2,q}_{\delta-2}
    \]
    is continuous and Fredholm for all $j\in\Ints$ and each $1<q<\infty$. Moreover,
    \begin{enumerate}[(a)]
    \item $\ker \overline{P^{j,q}_\delta}$ is independent of $j$ and $q$ and consists of polynomials. Consequently it is trivial if $\delta<0$,
    \item\label{part:flat-lap-image} $\im \overline{P}^{j,q}_\delta = \left(\ker \overline{P}^{2-j,q'}_{2-n-\delta}\right)^\perp$ in the following sense: a vector $Z \in \im \overline{P}^{j,q}_\delta$ if and only if $\int_{\Reals^n} \left<Z,K\right>_g\; dV_g=0$ for all $K\in \ker \overline{P}^{2-j,q'}_{2-n-\delta}$,
    \item the Fredholm index $\iota\left(\overline{P}^{j,q}_\delta\right)$ is independent of $j$ and $q$,
    \item if $2-n<\delta<0$, then $\overline{P}^{j,q}_\delta$ is an isomorphism,
    \item for all $X\in W^{j,q}_\delta$ and each $\delta'\in\Reals$
    \begin{equation}\label{eqn:flat-semi-fred}
    ||X||_{W^{j,q}_\delta} \lesssim ||\overline{P} X||_{W^{j-2,q}_{\delta-2}} + ||X||_{W^{j-2,q}_{\delta'}}.
    \end{equation}
    \end{enumerate}
\end{proposition}
\begin{proof}
    Theorem 3 of \cite{lockhart_elliptic_1983} implies that if $j\ge 2$ and if $\delta$ is non-exceptional, then
    $\overline{P}^{j,q}_\delta$ is Fredholm.  Lemma \ref{lem:self-adjoint} shows that 
    $\left(\overline{P}^{j,q}_\delta\right)^* = \overline{P}^{2-j,q'}_{2-n-\delta}$. 
    Noting that $2-n-\delta$ is non-exceptional if and only if $\delta$ is,
    we conclude that $\overline{P}^{j,q}_{\delta}$ is Fredholm for $j\le 0$ as well
    since the adjoint of a Fredholm map is Fredholm. 

    Theorem 3 of \cite{lockhart_elliptic_1983} also implies that $\ker \overline{P}^{j,q}_\delta$
    consists of polynomials if $j\ge 2$. 
    Lemma \ref{lem:basic-reg} implies that this same fact is true for fixed $q$ but arbitrary $j$.  Moreover, a polynomial is contained in $W^{j,q}_{\delta}$ if and only if its order is less than $\delta$ and hence the kernel depends on $\delta$ but is independent of $j$ and $q$.
    We denote this common kernel by $\ker \overline{P}_\delta$.

    To handle the marginal case $j=1$ we first observe from Lemma \ref{lem:basic-reg}
    that $\ker \overline{P}^{j,q}_\delta$ is independent of $j$ and $q$.
    Since 
    $\ker \overline{P}^{1,q}_\delta = \ker \overline{P}_\delta$, it
    is finite dimensional, and elementary arguments imply that $\im \overline{P}^{1,q}_\delta \subset W^{-1,q}_{\delta-2}\cap (\ker \overline{P}_{2-n-\delta})^\perp $. If we show 
    the reverse containment we have established that $\overline{P}^{1,q}_\delta$ has a closed range and a finite dimensional cokernel,
    and hence $\overline{P}^{1,q}_{\delta}$ is Fredholm.
    To establish the reverse containment, 
    suppose $Z \in W^{-1,q}_{\delta-2} \cap (\ker \overline{P}^{-1,q'}_{-n-\delta})^\perp$. From our established results for $\overline{P}^{-2,q}_{\delta}$ there exists $X\in W^{0,q}_{\delta}$ with $\overline{P}X = Z$, and Lemma \ref{lem:basic-reg} implies that $X\in W^{1,q}_\delta$ as needed.

    Since the image of 
    $\overline{P}^{j,q}_\delta$ is characterized in terms of the kernel of the adjoint for a Fredholm map (this is item \eqref{part:flat-lap-image}),
    we see that the dimension of the cokernel is independent of $j$ and $q$, 
    and hence so is the index.

    The isomorphism range $2-n<\delta<0$ follows from the fact that 
    for these values of $\delta$ we have $\delta<0$ and $2-n-\delta<0$,
    so the kernels of both $\overline P^{j,q}_\delta$ and its
    adjoint $\overline{P}^{2-j,q'}_{2-n-\delta}$ are trivial. 

    Finally, inequality \eqref{eqn:flat-semi-fred} is a special case of
    a more general fact about Fredholm operators.  If $T:X\rightarrow Y$
    is a semi-Fredholm map between Banach spaces, and if $Q: X\to Z$  is a 
    continuous map that is injective on $\ker T$, then for all $x\in X$,
    \[
    ||x||_X \lesssim ||Tx||_Y + ||Q x||_Z.
    \]
\end{proof}

The next two results concern the asymptotics 
of solutions of $\overline PX = Z$.  The first shows that if $Z$ has $O(r^{\delta-2})$ decay then $\overline PX = Z$ is solvable, up to an error on a compact set, with $X$ having $O(r^{\delta})$ decay. The second
concerns the specific asymptotic structure of a solution of $PX = Z$ if $Z$ has compact support.

\begin{lemma}\label{lem:P-ext-exact}
Suppose $k\in\Ints$, $1<p<\infty$ and that $\delta$ is non-exceptional.
If $Z_a\in W^{k-2,p}_{\delta-2}$, then for any fixed $R>0$ 
there exists $X^a\in W^{k,p}_{\delta}$ such that 
$(\overline P X)_a = Z_a$ on $E_R$.
\end{lemma}
\begin{proof}
    We need only consider the case $\delta<2-n$, since $\overline P$ is 
    surjective otherwise.  
    
    Recall that the adjoint of $\overline{P}^{k,p}_{\delta}$ is 
    $\overline P^{2-k, p'}_{2-n-\delta}$.
    Letting $\{H^a_{(j)}\}_{j=1}^m$ be a basis
    for $\ker \overline P^{2-k, p'}_{2-n-\delta}$, Proposition \ref{prop:flatcase} implies that 
    the components of each $H^a_{(j)}$ are polynomials.
    
    Define $T:C^\infty_c(B_R)\to \Reals^m$ via $T(X)_{j} = \int_{B_R} H_{(j)}^a X_a \; dV$.
    This map is surjective, for otherwise there would exist a nonzero $\beta\in \Reals^m$ with $\beta\cdot T(X)=0$
    for all compactly supported smooth vector fields on $B_R$.  As a consequence $\beta^j H^a_{(j)}$
    vanishes on $B_R$. But since the elements of the kernel are polynomials,
    we find that $\beta^j H^a_{(j)}=0$ on $\Reals^n$, which contradicts the linear independence of the basis elements.

    Now fix $Z_a \in W^{k-2,p}_{\delta-2}$. From the surjectivity of $T$ we can find $Y_a\in C^\infty_c(B_R)$ with $T(Y)_{j} = \int_{\Reals^n} H^a_{(j)} Z_a \; d\overline{V}$
    for each $j$.  Then $Z_a - Y_a$ is $L^2$ orthogonal to the kernel of 
    the adjoint of $\overline P^{k,p}_\delta$ and Proposition \ref{prop:flatcase}
    implies that there exists $X^a\in W^{k,p}_\delta$ which solves $(\overline PX)_a = Z_a - Y_a$.
    Since $Y_a=0$ outside $B_R$, the proof is complete.
    \\
\end{proof}

% The momentum carrier vector fields $W^a_{(j)}$ of Proposition \ref{prop:Wdual} are constructed to have compactly supported image under the action of $P_{g,N}$. This result, Lemma \ref{lem:P-ext-exact}, allows one to construct a vector field agreeing with $W^a_{(j)}$ at leading order and having compactly supported image under $\overline P$; this is done in the proof of Proposition \ref{prop:Wdual}. We now establish that such vector fields are closely related to the Green's function $G^{ab}$ of equation \eqref{eqn:greens}, a necessary ingredient for explicitly characterizing the leading order behavior of the second fundamental form $\widetilde K_{ab}$ of our initial data sets, as achieved in Theorems \ref{thm:summary_ck} and \ref{thm:summary_a}:

We now establish multipole expansions for solutions of $\overline P X = Z$ if $Z$ has compact support. The monopole term in this expansion 
is used to construct the momentum carrier vector fields $W^a_{(j)}$ of Proposition \ref{prop:Wdual}.

\begin{lemma}\label{lem:P-multipolar}
Suppose $k\in\Ints$, $1<p<\infty$, and that $\delta < 0$ is non-exceptional. 
Let $X^a\in W^{k,p}_{\delta}$ be a vector field such that 
$(\overline P X)_a$ is supported on $B_R$ for some $R>0$ and fix $\ell\in\Nats$.  There exist constants $c_b^\alpha$ (depending on indices $1\le b\le n$ and multiindices $\alpha$ with $|\alpha|\le \ell$) along with a vector field $Y^a\in W^{k,p}_{(2-n-\ell)^+}$ such that on $E_{2R}$
\begin{equation}\label{eq:veclap-decomp}
X^a = \sum_{|\alpha|\le \ell} c_b^\alpha \partial_\alpha G^{ab} + Y^a.
\end{equation}
Here, $G^{ab}$ is the Green's function from equation \eqref{eqn:greens}.
\end{lemma}
\begin{proof}
Let $Z_a = (\overline P X)_a$ and let 
\[
U^a(x) = \int G^{ab}(x-y) Z_b(y) \; d\overline{V}(y)
\]
with the convolution meant in the sense of distributions 
if $k<2$. Then $\overline P U = Z$ and the proof proceeds by showing first that $U^a$ admits the decomposition \eqref{eq:veclap-decomp},
and then that $U^a=X^a$, which completes the proof.

For each multiindex $\alpha$ with $|\alpha|\le \ell$, define $c_b^\alpha := \int \frac{(-y)^\alpha}{\alpha !} Z_b(y)\; d\overline V(y)$; this quantity is well defined since $Z_b$ has compact support. Note that if $\alpha =(\alpha_1,\ldots,\alpha_n)$ then $z^\alpha = z_1^{\alpha_1}\cdots z_n^{\alpha_n}$ and $\alpha ! = \alpha_1!\cdots \alpha_n!$. We then write
\[
    \int G^{ab}(x-y) Z_b(y) \; d\overline{V}(y) =
    \sum_{|\alpha| \le \ell }  c_b^\alpha \partial_\alpha 
    G^{ab}(x) + W^a
\]
with
\[
W^a(x) = \int_{\Reals^n} H^{ab}(x,y) F_b(y)\; d\overline{V}(y),
\]
and with 
\[
H^{ab}(x,y)=G^{ab}(x-y) - \sum_{|\alpha|\le \ell} \frac{(-y)^\alpha}{\alpha !} \partial_\alpha G(x-y).
\]

On the region $|x|\ge 2R$, we find that $W^a$ is smooth
since it is the convolution of a smooth function with $Z_a$.  
Since $G^{ab}$ is smooth and homogeneous of degree ${2-n}$ on 
$\Reals^n\setminus \{0\}$,
if $|y|\le R$ and if $|x|\ge 2R$ then
\[
|H^{ab}(x,y)| \lesssim |x|^{2-n-\ell}|y|^{\ell}.
\]
For each multi-index $\alpha$, the derivative $\partial^\alpha G^{ab}$
is homogeneous of order $2-n-|\alpha|$, and the same argument
then shows that
\[
|\partial^{\alpha}_x(H^{ab}(x,y))| \lesssim |x|^{2-n-|\alpha|-\ell}|y|^\ell.
\]
Hence on the region $|x|\ge 2R$ we have the uniform estimates
$|\partial^\alpha W^a|\lesssim |x|^{2-n-|\alpha| -\ell}$.

Let $\chi$ be a cutoff function equal to $1$ on $B_R$ and vanishing on $E_{2R}$. 
We then have
\[
U^a =  (1-\chi)\left[ \sum_{|\alpha|\le \ell} c^\alpha_b \partial_\alpha G^{ab}\right] + \underbrace{\chi U^a + (1-\chi) W^a}_{=:Y^a}.
\]
The uniform decay estimates on the derivatives of $W^b$ and the compact
support of $\chi Z^a$ show that $Y^a$ belongs to
$W^{k,p}_{(2-n-\ell)^+}$. Hence $U^a$ admits the 
decomposition \eqref{eq:veclap-decomp}, and that same decomposition 
implies that $U^a\in W^{k,p}_{(2-n)^+}$. Since $\delta<0$ and since
$(\overline PX)_a = (\overline P Z)_a$, we conclude that $X^a = Z^a$, which completes the proof.
\\
\end{proof}

We now return to the general case of the vector Laplacian $P = P_{g,N}$. A standard perturbation technique (compare with \cite{bartnik1} Theorem 1.10) provides the following coercivity estimate, generalizing equation (\ref{eqn:flat-semi-fred}):

\begin{proposition}\label{prop:semifred}
    Suppose that $g_{ab}$ and $N$ satisfy Assumption \ref{assumption:appendix} with parameters $k$, $p$, and $\tau$.
    Suppose $(j,q)\in \mathcal{S}^{k,p}$, and that $\delta$
    is non-exceptional.  Then for all $X\in W^{j,q}_\delta$ and any $\delta'\in\Reals$,
    \begin{equation}\label{eqn:semifred}
        ||X||_{W^{j,q}_\delta} \lesssim ||P X||_{W^{j-2,q}_{\delta-2}} + ||X||_{W^{j-2,q}_{\delta'}}.
    \end{equation}
\end{proposition}    
\begin{proof}
Let $X\in W^{j,q}_\delta$. 
Lemma \ref{lem:basic-reg} implies that 
\be \label{eqn:regularity-comp}
\|X\|_{W^{j,q}_\delta} \lesssim \| P X \|_{W^{j-2,q}_{\delta-2}} + \|X\|_{W^{j-2,q}_\delta}
\ee
and we wish to improve the final term from $\|X\|_{W^{j-2,q}_\delta}$ to $\|X\|_{W^{j-2,q}_{\delta'}}$.

Let $\chi(x)$ be a decreasing smooth function vanishing for $|x| > 2$ and satisfying $\chi(x) = 1$ for $|x| < 1$, and set $\chi_R(x) := \chi(x/R)$ with $R > 0$ sufficiently large chosen below. Decompose $ X = X_0 + X_\infty$, with $X_0 := \chi_R X$ and $X_\infty := (1- \chi_R) X$. Since $X_0$ is supported in a ball, we have that
\begin{align}
\|X\|_{W^{j-2,q}_\delta} & \leq \|X_0\|_{W^{j-2,q}_\delta} + \|X_\infty\|_{W^{j-2,q}_\delta}  \notag
\\ & \lesssim \|X\|_{W^{j-2,q}_{\delta'}} + \| X_\infty\|_{W^{j,q}_\delta} \label{eqn:vector-split}
\end{align}
where we have also used a trivial estimate for $X_\infty$.
Inequality \eqref{eqn:flat-semi-fred} of Proposition \ref{prop:flatcase} implies that 
\begin{align}
\|X_\infty \|_{W^{j,q,}_\delta} &\lesssim \|\overline P X_\infty \|_{W^{j-2,q}_{\delta-2}} +  
\| X_\infty \|_{W^{j-2,q}_{\delta'}}\nonumber\\
&\leq \| P X_\infty \|_{W^{j-2,q}_{\delta-2}} + \|(P - \overline P) X_\infty \|_{W^{j-2,q}_{\delta-2}}
+  \| X_\infty \|_{W^{j-2,q}_{\delta'}}.\label{eq:X-inf-bound1}
\end{align}
We observe that $(P - \overline P) X_\infty = \chi_{R/2} (P - \overline P) X_\infty$.  We then can write
\[
(P - \overline P) = \sum_{|\alpha|\le 2} A^\alpha \partial_\alpha
\]
with the matrix coefficients $A^\alpha \in W^{k-2+|\alpha|,p}_{\tau-2+|\alpha|}$, and Lemma \ref{lem:vanish} proved below implies that $\chi_{R/2} A^\alpha\to 0$
in $W^{k-2+|\alpha|,p}_{\tau-2+|\alpha|}$ as $R\to\infty$.  Lemma
\ref{lem:mult} then implies that $\chi_{R/2} (P - \overline P)$ converges
to zero in operator norm on $W^{j,q}_\delta$. Hence we can take $R$ sufficiently large to absorb the term $\|(P - \overline P) X_\infty \|_{W^{j-2,q}_{\delta-2}}$ into the left-hand inequality \eqref{eq:X-inf-bound1} to find
\[
\|X_\infty \|_{W^{j,q,}_\delta} \lesssim 
\| P X_\infty \|_{W^{j-2,q}_{\delta-2}} + \| X \|_{W^{j-2,q}_{\delta'}}.
\]
Now decompose $P X_\infty = \chi_R P X + [P, \chi_R] X$ and note that the coefficients of $[P, \chi_R]$ are supported on $B_R$. A computation using the bounded domain analog of Lemma \ref{lem:mult} (see \cite{holst_scaling_2023} Theorem 2.5) shows that $[P, \chi_R]$ is continuous as a map $W^{j,q}(B_R)\to W^{j-1,q}(B_R)$. 
This fact and the estimate $\| \chi_R P X \|_{W^{j-2,q}_{\delta-2}}\lesssim \| P X \|_{W^{j-2,q}_{\delta-2}}$ imply that
\be\label{eq:X-inf-est}
\|X_\infty \|_{W^{j,q,}_\delta} 
\lesssim 
\| P X \|_{W^{j-2,q}_{\delta-2}} + \| X \|_{W^{j-1,q}(B_R)} + \| X \|_{W^{j-2,q}_{\delta'}}.
\ee

For each $\epsilon > 0$, Sobolev interpolation implies that there is a $C_\epsilon > 0$ such that
\be \label{eqn:interpolation}
\|X\|_{W^{j-1,q}(B_R)} \leq \epsilon \|X\|_{W^{j,q}(B_R)} + C_\epsilon \|X\|_{W^{j-2,q}(B_R)} \lesssim \epsilon \|X\|_{W^{j,q}_\delta} + C_\epsilon \|X\|_{W^{j-2,q}_{\delta'}}
\ee
where the implicit constant is independent of $\epsilon$. 
Combining inequalities \eqref{eqn:regularity-comp}-\eqref{eqn:interpolation} and choosing $\epsilon$ sufficiently small
so as to absorb the term $\epsilon \|X\|_{W^{j,q}_\delta}$ into
the left-hand side, we obtain the desired inequality:
\[
\|X\|_{W^{j,q}_\delta} \lesssim \|P X\|_{W^{j-2,q}_{\delta-2}} + \|X\|_{W^{j-2,q}_{\delta'}}.
\]
\end{proof}

It remains to prove the following technical lemma used in the proof of Proposition \ref{prop:semifred}.
\begin{lemma}\label{lem:vanish}
Let $\chi$ be an increasing smooth function that equals zero on $B_{1}$ and equals $1$ on $E_2$, and define $\chi_R(x)=\chi(x/R)$. For any $u\in W^{j,q}_\delta$,
\[
\lim_{R\to\infty}||\chi_R u||_{W^{j,q}_\delta} = 0.
\]
\end{lemma}
\begin{proof}
Consider a scaling operator $S_r$.  Then $S_r^* (\chi_R u)$ vanishes on $A_1$ unless $r\ge R$.  Moreover, it equals $S_r^* u$ on $A_1$ unless $1\le r/R \le 2$, in which case the values and the derivatives of $S_r^*(\chi_R)$ are uniformly bounded independent of $r$ and $R$.  Hence for $R\ge 1$,
\[
||\chi_R u||_{W^{j,q}_\delta}^q \lesssim \sum_{\substack{j\in\Nats\\ 2^j \ge R}} 2^{-j \delta q} ||S_{2^j}u||^q_{W^{j,q}(A_1)}.
\]
Taking the limit in $R$ proves the result.
\\
\end{proof}

Combining estimate \eqref{eqn:semifred} from Proposition \ref{prop:semifred} with the compact embedding 
$W^{j,q}_{\delta} \hookrightarrow W^{j-2,q}_{\delta'}$ for $\delta'>\delta$,
 a standard argument (e.g., Theorem 1.10 of \cite{bartnik1})
implies that $P$ is semi-Fredholm:

\begin{corollary} \label{cor:semifred}
    Suppose that $g_{ab}$ and $N$ satisfy Assumption \ref{assumption:appendix} with the parameters $k$, $p$, and $\tau$.
    Suppose $(j,q)\in \mathcal{S}^{k,p}$, and that $\delta$
    is non-exceptional. Then
    \[
    P : W^{j,q}_\delta \rightarrow W^{j-2,q}_{\delta-2}
    \]
    is semi-Fredholm (that is, it has a finite dimensional kernel and a closed range).
\end{corollary}
\qed

As a final step before establishing the full Fredholm theory,
we now show that the kernel of $P$ depends on $\delta$ 
but not on the Sobolev parameters. This fact is not immediate 
from Lemma \eqref{lem:basic-reg}, which allows for an improvement of the number of derivatives, but only for a fixed integrability parameter $q$.  

\begin{lemma}\label{lem:kernel}
    Suppose that $g_{ab}$ and $N$ satisfy Assumption \ref{assumption:appendix} with the parameters $k$, $p$, and $\tau$.
    If $\delta$ is non-exceptional and if $(j,q)\in \mathcal S^{k,p}$
    from Definition \ref{def:calS} then the kernel of $P$ acting on $W^{j,q}_\delta$ agrees with
    the kernel of $P$ acting on $W^{k,p}_\delta$.  Moreover, 
    the kernel of $P$ acting on $W^{k,p}_\delta$ 
    is the same as the kernel of $P$ acting on $W^{k,p}_{\delta'}$ 
    for all\footnote{$\lfloor \cdot \rfloor$ denotes the floor operation, which rounds its input down to the nearest equal or smaller integer.} $\lfloor \delta \rfloor<\delta'\le \delta$.
\end{lemma}
\begin{proof}
Suppose $K\in W^{j,q}_\delta$ lies in the kernel of $P$.
We first show that $K\in W^{j,q}_{\delta'}$ for some $\delta'<\delta$.

Indeed,
\[
\overline P K = (\overline P - P) K + PK = (\overline P - P) K.
\]
The decay of the coefficients of $(\overline P - P)$ then implies that 
$\overline P K \in W^{j-2,q}_{-2+\delta+\tau}$.  Now pick $\delta'> \lfloor \delta \rfloor$ which satisfies $\delta+\tau<\delta'<\delta$. 
Because the interval $[\delta',\delta]$ contains no integers,
the obstruction conditions of Proposition \ref{prop:flatcase}  
part \eqref{part:flat-lap-image} are the same for $\delta$ 
and $\delta'$ and hence there is a 
$Y\in W^{j,q}_{\delta'}$ with $\overline P K = \overline P Y$.
We conclude that $Y$ and $K$ differ at most by a polynomial of degree less than $\delta'$ and consequently $K\in W^{j,q}_{\delta'}$ as well.

Now suppose that $j<k$.  If $(j+1,q)\in \mathcal S^{k,p}$ we can apply
Lemma \ref{lem:basic-reg} to conclude that $K\in W^{j+1,q}_\delta$.
Otherwise inequality \eqref{eqn:lap-condition1} 
shows that we can lower $q$ to some $Q$ so that $(j+1,Q)\in \mathcal S^{k,p}$. Although $W^{k,q}_\delta$ does not embed in $W^{k,Q}_\delta$
directly, we can use the faster decay rate and the 
embedding $W^{j,q}_{\delta'}\hookrightarrow W^{j,Q}_{\delta}$ 
along with Lemma \ref{lem:basic-reg} to conclude that $K\in W^{j+1,Q}_{\delta}$. After finitely many iterations of this process, we conclude that $K\in W^{k,Q}_\delta$ for some $Q$ with $(k,Q)\in \mathcal S^{k,p}$.
We now wish to improve $Q$ to $p$.

In fact, from our arguments above we know that  
$K\in W^{k,Q}_{\delta'}$ for some $\delta'< \delta$.  
But $W^{k,Q}_{\delta'}\hookrightarrow W^{k-1,Q'}_{\delta}$
for all $Q'<\infty$ with $1/Q' \ge 1/Q - 1/n$.  
If $1/Q-1/n> 1/p$ then we can take $1/Q' = 1/Q-1/n$
and apply Lemma \ref{lem:basic-reg} to conclude that $K\in W^{k,Q'}_\delta$.
After finitely many iterations we find that $K\in W^{k-1,Q''}_{\delta'}$
with $Q'' \ge p$ and $\delta'<\delta$.  But then $K\in W^{k-1,p}_\delta$
as well and a last application of Lemma \ref{lem:basic-reg} shows that 
$K\in W^{k,p}_{\delta}$. 

Finally, repeating the decay lowering argument from the start of the proof finitely many times depending on the size of $\tau$ 
we conclude that $K\in W^{k,p}_{\delta'}$ for all
$\lfloor \delta \rfloor<\delta'\le \delta$.
\\
\end{proof}

We now obtain the primary Fredholm resul:

\begin{proposition}\label{prop:fredholm}
    Suppose that $g_{ab}$ and $N$ satisfy Assumption \ref{assumption:appendix} with the parameters $k$, $p$, and $\tau$.
    Suppose $(j,q)\in \mathcal{S}^{k,p}$ and $\delta$
    is non-exceptional.  Then
    \[
    P : W^{j,q}_{\delta} \rightarrow W^{j-2,q}_{\delta-2}
    \]
    is Fredholm, and its Fredholm index satisfies $\iota(P^{j,q}_\delta) = \iota(\overline P^{j,q}_\delta)$ and
    is independent of $j$ and $q$.
    
    Let $P^*$ denote the action of $P$ on
     $W^{2-j,q'}_{2-n-\delta}$.  Then the kernel of $P^*$ is independent of $j$ and $q$ and depends only on $\delta$. For any $V \in W^{j-2,q}_{\delta-2}$, the equation
    \[
    PX = V
    \]
    is solvable for $X \in W^{j,q}_\delta$ if and only if $\int_{\Reals^n} \left<V, K \right>_g dV_g = 0$ 
    for all vector fields $K \in \ker(P^*)$.
\end{proposition}

\begin{proof}
     Corollary \ref{cor:semifred} implies that $P$ is semi-Fredholm.
     Lemma \ref{lem:self-adjoint} shows that the adjoint of $P$ is
     $P: W^{2-j,q'}_{2-n-\delta} \to W^{-j,q'}_{-n-\delta}$, which we denote by $P^*$.
     Since $2-n-\delta$ is non-exceptional if $\delta$ is,    Corollary \ref{cor:semifred} implies that $P^*$
    is also semi-Fredholm.  We claim the following:
    \be \label{eqn:kernelperp}
    \mathrm{im}(P) = (\ker(P^*))^\perp.
    \ee
    Together with Lemma \ref{lem:dual}, this establishes the claimed solvability criterion. 
    Since the kernel of $P^*$ is finite dimensional, this also establishes that $P$ is Fredholm. Lemma \ref{lem:kernel}  
    shows that the kernel of $P^*$ is independent of $j$ and $q$.
    
    Lemma \ref{lem:self-adjoint} implies that 
    \[
    \int_{\Reals^n} \left< PX, Y\right>_g dV_g = \int_{\Reals^n} 
    \left<X, P^* Y\right>_g dV_g
    \]
    for all $X \in W^{j,q}_{\delta}$ and $Y \in W^{2-j,q'}_{2-n-\delta}$, and the inclusion $\mathrm{im} \, P \subset (\ker P^*)^\perp$ is
    immediate.  Conversely, suppose that $Z\in W^{j-2,q}_{\delta-2}$
    is not in the image of $P$.  Since $P$ is semi-Fredholm its image is closed and
    the Hahn-Banach Theorem together with Lemma \ref{lem:dual} imply that
    there exists a $Y\in W^{2-j,q'}_{2-n-\delta}$ satisfying both
    \be \label{eqn:Yperp}
    \int_{\Reals^n} \left<P X, Y\right>_g\; dV_g 
    =0
    \ee
    for all $X \in W^{j,q}_{\delta}$, and
    \begin{equation}\label{eqn:Ynotperp}
    \int_{\Reals^n} \left<Z, Y\right>_g\; dV_g \neq 0.
    \end{equation}
    Equation (\ref{eqn:Yperp}) implies that
    \[
    \int_{\Reals^n} \left< X, P^* Y\right>_g\; dV_g 
    =0
    \]
    for all $X\in W^{j,q}_\delta$, so Lemma \ref{lem:dual} further yields that $Y\in \ker P^*$.
    That is, if $Z$ is not in the image of $P$, then 
    there exists $Y\in \ker(P^*)$ which satisfies equation \eqref{eqn:Ynotperp},
    and hence $Z\not \in (\ker(P^*))^\perp$. This confirms equation \eqref{eqn:kernelperp}.
    
    Regarding the index, the paths of metrics $g_t = (1-t)\overline{g} + t g$
    and lapses $N_t = (1-t)+tN$ yield a continuous path of 
    operators $P_{g_t,N_t}$ from $\overline{P}$
    to $P = P_{g,N}$, and local constancy of the index shows
    that
    $\iota(P^{j,q}_\delta) = \iota(\overline P^{j,q}_\delta)$.
    The fact that the index is independent of $(j,q)$
    follows from the corresponding fact for $\overline P$ from 
    Proposition \ref{prop:flatcase}.
    \\
\end{proof}


\begin{thebibliography}{99} 

\bibitem{allen_sobolevclass_2022a}  P. T. Allen, J. M., Lee, \& D. Maxwell. 
  \begin{itshape} Sobolev-class Asymptotically Hyperbolic Manifolds and the Yamabe Problem. \end{itshape} 
   arXiv:2206.12854. (2022).

\bibitem{lz1} Z. An, L. Bieri. 
\begin{itshape} Null Limits and Antipodal Symmetries of Dynamical Spacetimes. \end{itshape} 
In preparation. (2025). 

\bibitem{bartnik1} R. Bartnik. 
  \begin{itshape} The Mass of an Asymptotically Flat Manifold. \end{itshape} 
  Communications on Pure and Applied Mathematics. Vol. XXXIX 661-693. (1986).

  \bibitem{beig1997tt} R. Beig
	\begin{itshape} TT-tensors and Conformally Flat Structures on 3-Manifolds.  \end{itshape}
Banach Center Publications, 41(1), 109-118. (1997).

\bibitem{beig_momentum_1996} R. Beig, N. O. Murchadha. 
  \begin{itshape} The Momentum Constraints of General Relativity and Spatial Conformal Isometries. \end{itshape} 
  Communications in Mathematical Physics. 176(3). 723-738. (1996).

\bibitem{lydia4} L. Bieri.  
        \begin{itshape}  New Structures in Gravitational Radiation.     \end{itshape}
Advances in Theoretical and Mathematical Physics. 26. 3. 531-594. 
arXiv: 2010.07418. (2022).

\bibitem{lydia14} L. Bieri.    
 {\itshape New Effects in Gravitational Waves and Memory.} 
Phys. Rev. D 103. 024043.
arXiv: 2010.09207. (2021).

\bibitem{lydia5} L. Bieri.  
        \begin{itshape}  Radiation and Asymptotics for Spacetimes with Non-Isotropic Mass.     \end{itshape}
Pure and Applied Mathematics Quarterly. Vol. 20. Number 4. (2024).
 

\bibitem{lydia1} L. Bieri.  
        \begin{itshape} An Extension of the Stability Theorem of the Minkowski Space
in General Relativity. \end{itshape}
        ETH Zurich, Ph.D. thesis. 17178. 
        Zurich. (2007).  
        
\bibitem{lydia2} L. Bieri.  
        \begin{itshape} Extensions of the Stability Theorem of the Minkowski Space
in General Relativity. Solutions of the Einstein Vacuum Equations. \end{itshape}
        AMS-IP. Studies in Advanced Mathematics. Cambridge. MA. (2009).         

\bibitem{bieri_brill_2025} L. Bieri, D. Garfinkle, J. Wheeler.  
        \begin{itshape} Brill Waves With Slow Fall-off Towards Spatial Infinity. \end{itshape}
        Classical and Quantum Gravity. Volume 42. 11. (2025).

\bibitem{bowen_york} J.M. Bowen, J.W. York Jr.  
        \begin{itshape} Time-asymmetric initial data for black holes and black-hole collisions. \end{itshape}
        Physical Review D, 21(8), 2047. (1980).

\bibitem{Brill1} D. Brill. 
  \begin{itshape} On the Positive Definite Mass of the Bondi-Weber-Wheeler Time-Symmetric Gravitational Waves.   \end{itshape} 
Annals of Physics. 7. 466-483. (1959). 

\bibitem{choquet1952} Y. Choquet-Bruhat.
  \begin{itshape} Théorème d'Existence Pour Certains Systèmes d'\'Equations aux Dérivées Partielles Non Linéaires.   \end{itshape} 
Acta Mathematica, 88, 141-225. (1952).

\bibitem{choquet1969global} Y. Choquet-Bruhat, R. Geroch.
  \begin{itshape} Global Aspects of the Cauchy Problem in General Relativity.   \end{itshape} 
Communications in Mathematical Physics, 14(4), 329-335. (1969).

\bibitem{bruhat-isenberg} Y. Choquet-Bruhat, J. Isenberg, \& J.W. York Jr.
  \begin{itshape} Einstein Constraints on Asymptotically Euclidean Manifolds.   \end{itshape} 
Physical Review D, 61(8), 084034. (2000).

\bibitem{bruhat-york} Y. Choquet-Bruhat, J. W. York Jr.
  \begin{itshape} The Cauchy Problem.   \end{itshape} 
General relativity and gravitation: one hundred years after the birth of Albert Einstein, edited by A.\@ Held, 1, 99-172. (1980).

\bibitem{sta} D. Christodoulou, S. Klainerman.
        \begin{itshape} The Global Nonlinear Stability of the Minkowski Space.
\end{itshape}
        Princeton Math. Series. 41. 
        Princeton University Press. Princeton. NJ. (1993).

\bibitem{dain2001asymptotically} S. Dain, H. Friedrich.
        \begin{itshape} Asymptotically Flat Initial Data with Prescribed Regularity at Infinity.
\end{itshape}
        Communications in Mathematical Physics, 222(3), 569-609. (2001). 
        
\bibitem{holst_scaling_2023} M. Holst, D. Maxwell, \& G. Tsogtgerel.
    \begin{itshape} A Scaling Approach to Elliptic Theory for Geometrically-Natural Differential Operators with Sobolev-Type Coefficients.
\end{itshape}
        arXiv:2306.15842. (2023).
        
\bibitem{Jackson} J. D. Jackson
    \begin{itshape} Classical Electrodynamics
    \end{itshape}
(third edition). Wiley. (1999).

\bibitem{lichnerowicz1944} A. Lichnerowicz
  \begin{itshape} L’int\'egration des \'Equations de la Gravitation Relativiste et la Probl\`eme des n Corps.   \end{itshape} 
Journal de Mathematiques Pures et Appliquees, 23, 37-63. 1944.
        
\bibitem{lockhart_elliptic_1983} R. B. Lockhart, R. C. Mc Owen.
    \begin{itshape} Elliptic Differential Operators on Noncompact Manifolds.
    \end{itshape}
Annali della Scuola Normale Superiore di Pisa-Classe di Scienze, 12(3), 409-447. (1985).
        
\bibitem{symmmagdyetal} M. Magdy Ali Mohamed et al. 
	\begin{itshape} BMS-Supertranslation Charges at the Critical Sets of Null Infinity. \end{itshape} 
Journal of Mathematical Physics, 65, 032501. (2024). 

\bibitem{maxwell_rough_2006} D. Maxwell.
	\begin{itshape} Rough Solutions of the Einstein Constraint Equations. 
    \end{itshape} 
    Journal f\"ur die reine und angewandte Mathematik. 590. 1--29. (2006).


\bibitem{maxwell_solutions_2005} D. Maxwell.
        \begin{itshape} Solutions of the Einstein Constraint Equations with Apparent Horizon Boundaries.
\end{itshape}
        Communications in Mathematical Physics. 253(3). 561-583. (2005).

\bibitem{mitrea} D. Mitrea,
        \begin{itshape} Distributions, Partial Differential Equations, and Harmonic Analysis
\end{itshape}
        (second edition).
        Universitext.
        Springer. (2018).
        
\bibitem{symmprabhuetal} K. Prabhu et al.
	\begin{itshape} Infrared Finite Scattering Theory in Quantum Field Theory and Quantum Gravity.   \end{itshape}
Physical Review D, 106, 066005. (2022). 

\bibitem{strominger} A. Strominger.
\begin{itshape}
    On BMS Invariance of Gravitational Scattering.
\end{itshape}
Journal of High Energy Physics, 7, 152. (2014).

\bibitem{tafel2018all} J. Tafel.
	\begin{itshape} All Transverse and TT Tensors in Flat Spaces of any Dimension.  \end{itshape}
General Relativity and Gravitation, 50(3), 31. (2018).

\bibitem{triebel_spaces_1976a} H. Triebel.
        \begin{itshape} Spaces of Kudrjavcev Type I. Interpolation, Embedding, and Structure.
\end{itshape}
        Journal of Mathematical Analysis and Applications, 56(2), 253-277. (1976).

\end{thebibliography}
\end{document}